\documentclass[10pt,final,journal]{IEEEtran}

\IEEEoverridecommandlockouts

\usepackage[utf8]{inputenc}
\usepackage{bm}
\usepackage{amsmath}
\usepackage{amsthm}
\usepackage{amsfonts}
\usepackage{mathrsfs}
\usepackage{pifont}
\usepackage{amssymb}
\usepackage{verbatim}
\usepackage{upgreek}
\usepackage{color}
\usepackage[final]{graphicx}
\usepackage{caption} 
\usepackage{subfigure}

\usepackage{cuted}
\usepackage{afterpage} 
 
\usepackage{enumerate}
\usepackage{epsfig}
\usepackage{wrapfig}
\usepackage{epstopdf}
\usepackage{cite}
\usepackage{setspace}
\usepackage{makecell}

\usepackage{xcolor}
\usepackage{tcolorbox}

\usepackage{hyperref}
\hypersetup{hypertex=true, %
	colorlinks=true,
	linkcolor=red,
	anchorcolor=blue,
	citecolor=green}
\allowdisplaybreaks[2]
\usepackage{xr} 

\usepackage{pifont}

\usepackage{upgreek}

\usepackage{booktabs}
\usepackage{multirow}
\usepackage{tabularx} 

\usepackage{tabularx} 
\makeatletter
\newcommand{\multiline}[1]{%
	\begin{tabularx}{\dimexpr\linewidth-\ALG@thistlm}[t]{@{}X@{}}
		#1
	\end{tabularx}
}

\usepackage{cuted}

\usepackage[]{algpseudocode} 
\usepackage{algorithmicx,algorithm}
\algrenewcommand\alglinenumber[1]{ #1:}

\renewcommand{\algorithmicrequire}{{\tt \bf{Input:}}}
\renewcommand{\algorithmicensure}{{\tt \bf{Output:}}}

\newtheoremstyle{note}
{3pt}
{3pt}
{\upshape}
{}
{\bfseries}
{:}
{.5em}
{}
\theoremstyle{note}

\newtheorem{corollary}{Corollary}

\newtheorem{lemma}{Lemma}
\newtheorem{proposition}{Proposition}
\newtheorem{remark}{Remark}
\theoremstyle{plain}

\newtheoremstyle{noparentheses}
{\topsep}   
{\topsep}   
{\upshape}  
{0pt}       
{\bfseries} 
{.}         
{5pt plus 1pt minus 1pt} 
{\thmname{#1} \thmnumber{#2} \thmnote{#3}}  
\theoremstyle{noparentheses}

\newtheorem{theorem*}{Theorem}
\newtheorem{lemma*}{Lemma}
\newtheorem{definition*}{Definition}
\newtheorem{corollary*}{Corollary}
\newtheorem{proposition*}{Proposition}
\newtheorem{property*}{Property}

\newcommand{\epc}{\hspace{1pc}}

\newcommand{\Diag}{\mbox{\rm Diag}\, }

\newcommand{\vect}{\mbox{\rm vec}\, }
\newcommand{\ex}{{\mathbb E}}

\newcommand{\subto}{{\rm s.t.} }
\newcommand*{\trans}{{\mathsf{T}}}

\newcommand*{\ctrans}{{\mathsf{H}}}

\newcommand{\beq}{\begin{equation}}
	\newcommand{\eeq}{\end{equation}}
\newcommand{\bea}{\begin{array}*{20}{lll}}
	\newcommand{\eea}{\end{array}}
\newcommand{\disps}{\displaystyle}

\newcommand{\sss}{\scriptscriptstyle}

\newcommand{\upsmile}{\buildrel{\lower3pt\hbox{$\scriptscriptstyle\smile$}} 
\over}

\DeclareMathAlphabet\mathbfcal{OMS}{cmsy}{b}{n}

\def\LRT#1#2{\!
	\raisebox{.2ex}{$
		{{\scriptstyle\;#1}\atop{\displaystyle\gtrless}}
		\atop
		{\raisebox{-1.25ex}{$\scriptstyle\;#2$}}
		$}
	\!}

\newcommand{\bzero}{\mbox{\boldmath{$0$}}}

\newcommand{\bA}{\mbox{\boldmath{$A$}}}
\newcommand{\ba}{\mbox{\boldmath{$a$}}}
\newcommand{\bB}{\mbox{\boldmath{$B$}}}
\newcommand{\bb}{\mbox{\boldmath{$b$}}}

\newcommand{\bc}{\mbox{\boldmath{$c$}}}
\newcommand{\bD}{\mbox{\boldmath{$D$}}}
\newcommand{\bd}{\mbox{\boldmath{$d$}}}

\newcommand{\bh}{\mbox{\boldmath{$h$}}}
\newcommand{\bI}{\mbox{\boldmath{$I$}}}

\newcommand{\bM}{\mbox{\boldmath{$M$}}}

\newcommand{\bR}{\mbox{\boldmath{$R$}}}
\newcommand{\br}{\mbox{\boldmath{$r$}}}

\newcommand{\bs}{\mbox{\boldmath{$s$}}}

\newcommand{\bU}{\mbox{\boldmath{$U$}}}
\newcommand{\bu}{\mbox{\boldmath{$u$}}}

\newcommand{\bv}{\mbox{\boldmath{$v$}}}

\newcommand{\bx}{\mbox{\boldmath{$x$}}}

\newcommand{\by}{\mbox{\boldmath{$y$}}}

\newcommand{\bz}{\mbox{\boldmath{$z$}}}

\newcommand{\rmd}{{\rm d}}

\newcommand{\rmH}{{\rm H}}

\newcommand{\rmI}{{\rm I}}

\newcommand{\rmp}{{\rm p}}

\newcommand{\rmr}{{\rm r}}
\newcommand{\rmS}{{\rm S}}

\newcommand{\rmt}{{\rm t}}

\newcommand{\rmv}{{\rm v}}

\newcommand{\calF}{{\cal F}}
\newcommand{\calG}{{\cal G}}

\newcommand{\calO}{{\cal O}}
\newcommand{\calP}{{\cal P}}

\newcommand{\calS}{{\cal S}}

\newcommand{\bbC}{{\mathbb C}}
\newcommand{\bbR}{{\mathbb R}}

\newcommand{\bbH}{{\mathbb H}}

\newcommand{\balpha}{\mbox{\boldmath{$\alpha$}}}

\newcommand{\bgamma}{\mbox{\boldmath{$\gamma$}}}

\newcommand{\bPhi}{\mbox{\bm{{$\mit \Phi$}}}}

\newcommand{\bSigma}{\mbox{\bm{{$\mit \Sigma$}}}}
\newcommand{\bPsi}{\mbox{\bm{{$\mit \Psi$}}}}

\newcommand{\bmu}{\mbox{\boldmath{$\mu$}}}

\RequirePackage[normalem]{ulem} 
\RequirePackage{color}\definecolor{RED}{rgb}{1,0,0}\definecolor{BLUE}{rgb}{0,0,1} 

\usepackage{xr}
\begin{document}
	\title{Radar Code Design for the Joint Optimization of Detection Performance and Measurement Accuracy in Track Maintenance\\
	\thanks{The work of Augusto~Aubry, Vincenzo~Carotenuto, and Antonio~De~Maio was supported in part by
		the European Union under the Italian National Recovery and Resilience Plan
		(NRRP) of NextGenerationEU, partnership on “Telecommunications of the Future” (PE00000001 - Program “RESTART”). The work of Tao~Fan, Xianxiang~Yu, and Guolong~Cui was supported in part by the National Natural Science Foundation of China under Grants 62271126 and 62101097, and in part by the 111 Project under Grant B17008. (\emph{Corresponding author: Antonio~De~Maio}.)}
		\thanks{Tao~Fan, Xianxiang~Yu, and Guolong~Cui are with the School of Information and Communication Engineering, University of Electronic Science and Technology of China, Chengdu 611731, China. Tao~Fan is also with the Department of Electrical and Information Technology Engineering, Universit\`a degli Studi di Napoli “Federico II”, I-80125 Napoli, Italy. (e-mail: thaumielfan@gmail.com; xianxiangyu@uestc.edu.cn;  cuiguolong@uestc.edu.cn).}	
		\thanks{Augusto~Aubry, Vincenzo~Carotenuto, and Antonio~De~Maio are with the Department of Electrical and Information Technology Engineering, Universit\`a degli Studi di Napoli “Federico II”, I-80125 Napoli, Italy (e-mail: augusto.aubry@unina.it; vincenzo.carotenuto@unina.it; 
		ademaio@unina.it).}}
	
	\author{
		\IEEEauthorblockN{Tao~Fan, Augusto~Aubry, \emph{Senior Member, IEEE}, Vincenzo~Carotenuto, \emph{Senior Member, IEEE}, \\Antonio~De~Maio, \emph{Fellow, IEEE}, Xianxiang~Yu, \emph{Member, IEEE}, and Guolong~Cui, \emph{Senior Member, IEEE}}
	}
	
	\maketitle
	\IEEEpeerreviewmaketitle
	
		\label{abstract}
	\begin{abstract}
		\boldmath
		This paper deals with the design of slow-time coded waveforms which jointly optimize the detection probability and the measurements accuracy for track maintenance in the presence of colored Gaussian interference. The output signal-to-interference-plus-noise ratio (SINR) and Cram\'er Rao bounds (CRBs) on time delay and Doppler shift are used as {figures of merit} to accomplish reliable detection as well as accurate measurements. The transmitted code is subject to radar power budget requirements and a similarity constraint. To tackle the resulting non-convex multi-objective optimization problem, a polynomial-time algorithm that integrates scalarization and tensor-based relaxation methods is developed. The corresponding relaxed multi-linear problems are solved by means of the maximum block improvement (MBI) framework, where the optimal solution at each iteration is obtained in closed form. Numeral results demonstrate the trade-off between the detection and the estimation performance, along with the acceptable Doppler robustness achieved by the proposed algorithm.  
	\end{abstract}

	
	\section{Introduction}
	The advancements in flexible digital waveform generation technologies and high-speed signal processing equipments, have empowered radar systems to improve sensing performance by dynamically customizing waveforms to match the target and environment. Accordingly, numerous techniques have been proposed to address waveform optimization tailored for tasks involving detection, tracking, classification, and imaging\cite{gini2012waveform,farina2017impact, cui2020radar}.
	
	In the context of target tracking, bespoke waveforms for subsequent transmissions are designed to maintain a high-quality trajectory, leveraging the \emph{a priori} knowledge provided by the tracker regarding the target state and the sensing environment awareness \cite{haykin2006cognitive}. Two research streams have emerged. The former focuses on selecting the optimal waveform from a finite library of pre-configured waveforms to achieve the most accurate estimate of the target state for the next time step in the filtering process\cite{kershaw1994optimal,kershaw1997waveform,sen2010ofdm,benavoli2019joint}. The latter aims at synthesizing signals which ensure reliable target detection and accurate measurements for the tracker, potentially improving the tracking performance\cite{li2006signal,de2009design,aubry2012cognitive,aubry2013knowledge,stoica2011optimization,soltanalian2013joint,forsythe2005waveform,li2008range,huleihel2013optimal}. This paper frames along the second research line.
	
	There has been a great amount of research activities in waveform synthesis for point-like target track maintenance. The majority of existing work focuses on a single function perspective, prioritizing either the detection performance or the parameter estimation accuracy. For example, considering detection performance, the output signal-to-interference-plus-noise ratio (SINR) is commonly regarded as a design criterion. In \cite{li2006signal}, the maximization of the SINR under energy and similarity constraints was employed to enhance the detection probability for a stationary point-like target embedded in a colored disturbance with known covariance matrix. In \cite{de2009design}, the slow-time phased code ensuring a certain degree of similarity to a prescribed radar sequence was devised for a target with known Doppler shift. Along this line, in \cite{aubry2012cognitive}, the transmit waveform and receive filter with the same constraints as in \cite{de2009design} were jointly designed in the presence of signal-dependent clutter, whereas in \cite{aubry2013knowledge} the constant modulus constraint is replaced by an energy constraint to obtain additional design  degrees of freedom. In \cite{stoica2011optimization} and \cite{soltanalian2013joint}, under the assumption of discrete uncorrelated interference, transceivers were synthesized according to the mean-square error criterion forcing the transmitted waveform to exhibit either a constant modulus or a peak-to-average power ratio (PAPR) constraint. Considering spectrum congested environments\cite{aubry2019multi, aubry2021experimental}, suitable controls on the waveform spectrum during the design process were also studied in \cite{kay2007optimal, aubry2015new, aubry2016forcing, aubry2020design, carotenuto2020assessing, yu2022multispectrally,yang2022multi}. With a special emphasis on transmitter/transceiver synthesis to improve target detectability with known target measurements, \cite{cui2014mimo, wu2017transmit, yu2020mimoclutter, yu2020mimomulti} proposed efficient design procedures for multiple-input multiple-output (MIMO) radar systems operating in complex environments. Some other waveform design methods to enhance the detection performance leveraging on information-theoretic measures can be found in\cite{de2007design, tang2010mimo, tang2015relative,liang2015waveform,naghsh2017information, tang2019spectrally,fan2021spectrally,fan2020min}. Furthermore, robust design frameworks using worst-case SINR or average SINR for imprecisely known target delays, Doppler shifts, and azimuths were pursued in \cite{de2010doppler, naghsh2013doppler, aubry2015optimizing, zhao2017maximin, cui2017space, yao2023robust, fan2024joint}. 
	
	In the context of target parameters estimation, to infer accurately delay and Doppler, various radar signals with near-thumbtack-like ambiguity functions, including the Costas frequency coding, phase-coded, and multi-carrier phase-coded, were constructed\cite{levanon2004radar}. Waveform optimization for power constrained MIMO radar using Cram\'er Rao bound (CRB), was considered in \cite{forsythe2005waveform}, seeking to improve the accuracy of single target azimuth estimation in both spatial and temporal white noise conditions. In \cite{li2008range}, this approach was extended to the case of multiple targets in the presence of spatially colored noise. Three design criteria were investigated, involving the minimization of the trace, the determinant, and the largest eigenvalue of the CRB matrix. In \cite{huleihel2013optimal}, assuming unit target detection probability, the authors optimized the Bayesian CRB or the Reuven-Messer bound for a multi-step cognitive MIMO radar system to get adaptive waveforms. Some notable research on waveform design for target estimation under quality of service (QoS) constraints in dual-functional radar-communication systems was pursued in \cite{bicua2019radar,liu2021cramer}. The above mentioned techniques prioritize either the detection or the estimation performance independently, potentially leading to an uncontrollable impact on the other quality parameter. In particular, an excessive emphasis on detection performance may degrade estimation accuracy, resulting in a poor measurement-to-track data association. Conversely, without monitoring the detection performance may render a missed detection, leaving no measurement available for the tracking.
	
	As a matter of fact, few works have addressed the problem of waveform design to the trade-off between the detection performance and the estimation accuracy. In \cite{de2008code,de2007code}, considering a signal-independent noise environment, the slow-time code was devised by optimizing the output SINR while controlling the achievable values for the Doppler estimation accuracy and the similarity with a reference sequence. Moreover, the authors extended the approach in \cite{de2009code} to spatial-temporal processing accounting for an extra constraint for the azimuth (or the normalized spatial frequency) estimation accuracy. In \cite{de2011pareto}, aiming at maximizing the SINR and minimizing the CRB on the Doppler shift estimation accuracy, the problem of Pareto-optimal waveform design with the same constraints as in \cite{de2009code} was tackled by resorting to the scalarization technique. In \cite{sen2013papr}, the SINR and CRBs on Doppler, azimuth, as well as target-scattering coefficient estimation accuracies, for an orthogonal frequency division multiplexing MIMO (OFDM-MIMO) radar,  were taken as figure of merits along with the PAPR constraint. The non-dominated sorting genetic algorithm II was used to solve the result multi-objective optimization problem. 
	
	In this paper, a novel waveform design strategy for track maintenance is proposed, with the goal of optimizing both the detection performance and the accuracy of time delay and Doppler shift estimation. Specifically, a single-input single-output (SISO) radar tracking a point-like target in a colored noise environment is considered. The system emits a slow-time coded (both in amplitude and in phase) pulse train with each pulse modulated by a linear chirp signal. {Unlike most of the previous works\cite{de2007code,de2008code,de2009code,de2011pareto} that do not account for the delay-Doppler estimation accuracies, the SINR and the determinant of the CRB on delay and Doppler estimation are used as performance metrics.} Together with the power budget requirement and the similarity with a reference sequence, the design of slow-time signal is formulated as a constrained multi-objective optimization problem. In order to solve it, the scalarization technique \cite{de2011pareto} is first utilized to transform the original problem into general quartic programs with quadratic constraints (an inhomogeneous quartic polynomial in the objective function). Then, leveraging the tensor-based relaxation method in conjunction with the maximum block improvement (MBI) framework \cite{aubry2013ambiguity,chen2022generalized}, the transmitted code is constructed from the best block produced by a linear-improvement subroutine. The multi-linear function for the real-valued complex inhomogeneous quartic polynomial are formally derived, and the monotonicity of the tensor-based relaxation method is strictly guaranteed. The computational complexity of the procedure increases linearly with the number of iterations and squarely with the code length. At the analysis stage, some numeral results are provided to assess the performance of the new encoding algorithm in terms of convergence, Pareto curves, and Doppler sensitivity.
	
	To the authors best knowledge, this is the first study of a slow-time code synthesis technique that jointly optimizes the detection probability and the delay-Doppler estimation accuracies under energy and similarity constraints for colored Gaussian disturbance. In this respect, this paper offers the following main contribution
	\begin{enumerate}[a)]
		\item the development of a polynomial iterative procedure that encompasses scalarization, tensor-based relaxation, and the MBI approaches to address the resulting multi-objective optimization problem, and the investigation of the convergence properties for the relaxed procedure;
		\item the construction of a fourth order conjugate super-symmetric tensor for the real-valued complex inhomogeneous quartic function involved in tensor-based relaxation;
		\item the presentation of numerical results aimed at validating the theoretical achievements.
	\end{enumerate}
	
	The rest of the paper is organized as follows. Section \ref{SecII} introduces the signal model for target detection and parameters estimation and defines the radar code design problem. In Section \ref{SecIII}, an iterative method based on the scalarization technique and tensor-based relaxation is developed to synthesize the radar code. Section \ref{SecIV} offers some case studies to verify the performance. Finally, in Section \ref{SecV}, conclusions and some possible future research lines are provided. 
	
	\begin{table}[htbp]
		\centering  
		\renewcommand{\arraystretch}{1.1}
		\caption{Notations}
		\begin{tabular}{p{0.46\textwidth}}
				\hline  
				\begin{itemize}
				\item Bold letters, e.g., $\ba$ (lower case), and $\bA$ (upper case) represent vector and matrix, respectively, where the $i$th element  of $\ba$ and the $(i,j)$th entry of $\bA$ are, respectively, represented as $\ba(i)$ and $\bA(i,j)$.
				\item $\bI$ and $\bzero$ denote, respectively, the identity matrix and the matrix with zero entries (their size is determined from the context). 
				\item $\|\ba\|$ denotes the Euclidean norm of $\ba$.
				\item $\Diag(\ba)$ indicates the diagonal matrix whose $i$-th diagonal element is the $i$-th entry of $\ba$.
				\item $\text{eig}(\bA)$, $\lambda_{\min}(\bA)$, and $\lambda_{\max}(\bA)$ are the eigenvalues, the smallest, and largest eigenvalue of $\bA \in \bbH^N$, respectively. 
				\item $\det (\bA)$ denotes the determinant of $\bA$. 
				\item $\vect(\bA)$ denotes the column vector obtained by stacking the columns of $\bA$.
				\item For any $\bA\in\bbH^N$, $\bA \succ 0, \bA \succeq 0$  and $\bA \prec 0$ indicate that the matrix $\bA$ is positive definite, positive semi-definite and negative-definite, respectively. 
				\item Letter $\jmath$ represents the imaginary unit (i.e., $\jmath$ = $\sqrt{-1}$).
				\item $(\cdot)^\trans$, $(\cdot)^{*}$ and $(\cdot)^\ctrans$ are the transpose, complex conjugate, and conjugate transpose operators, respectively.
				\item $\mathbb{R}^N$, $\mathbb{C}^N$, $\mathbb{C}^{N\times N}$, $\bbH^{N}$, $\bbH_{++}^{N}$ are, respectively, the set of $N$-dimensional vectors of real numbers, complex numbers, $N\times N$ matrices, $N\times N$ Hermitian matrices, and Hermitian positive definite matrices.
				\item $\odot$ and $\otimes$ denote, respectively, the Hadamard product and the Kronecker product.
				\item For any complex number $x$, $\Re\{x\}$, $\Im\{x\}$, $|x|$, and $\arg(x)$ denote, respectively, the real part, imaginary part, modulus, and phase.
				\item $\Pi_k(i_1,i_2,\cdots,i_d)$ is the set of all $k$-permutations without repetition of the indices $\{i_1,i_2,\cdots,i_d\}$.
				\item Given two sets $A$ and $B$, $A\setminus B$ denotes the set of all elements in $A$ that are not in $B$.
				\item The statistical expectation is indicated as $\ex[\cdot]$.
				\item $\frac{\partial f(x)}{\partial x}\!$ means the derivative of $f(x)$ with respect to $x$.
				\end{itemize}\\
				\hline
		\end{tabular}
	\end{table}
	
	\section{Problem Formulation}\label{SecII}
	Let us consider a mono-static radar system that emits a burst of $M$ coherent pulses, modulated according to a slow-time coding; each pulse is a linear chirp signal $s(t)$ with bandwidth $B$ and pulsewidth $T_{\rm p}$, i.e., at baseband,
	\beq
		s(t)=\text{exp}\bigg[\jmath \pi \frac{B}{T_\rmp}\Big(t-\frac{T_\rmp}{2}\Big)^2\bigg]\text{rect}\Big(\frac{t}{T_\rmp}-\frac{1}{2}\Big),
	\eeq
	where $\text{rect}(t)$ is a rectangular function that is 0 outside of the interval $[-0.5, 0.5]$ and 1 within it. The baseband transmitted signal is 
	\beq
	x(t)= \sum_{m=0}^{M-1}c(m) s(t-mT_\rmr),
	\eeq
	where $T_\rmr$ denotes the pulse repetition interval (PRI), and $\bc=[c(0), c(1), \cdots, c(M-1)]^\trans \in \mathbb{C}^M$ is  the slow-time code. 
	
	The signal backscattered by {a single moving} point-like target and received by the radar is \cite{de2008code,de2009design}
	\beq
	r(t)= \alpha x{(t-\tau)} \text{exp}{[\jmath 2 \pi  (f_\rmd+f_0) (t-\tau)]},
	\eeq 
	where $\alpha$ is a complex parameter accounting for the target reflectivity and channel propagation, $\tau $ and $f_\rmd$ denote the round-trip time delay and the target Doppler frequency, and $f_0$ stands for the radar carrier frequency.  This signal is downconverted and sampled with a temporal step-size $\Delta t$. Hence, ignoring the intra-pulse Doppler shift, the fast-slow-time observation vector can be expressed as 
	
	\beq \label{echo}
		\br=\alpha \cdot (\ba(f_{\rm d}) \! \odot \! \bc) \! \otimes \! {\bs(\tau)} +   \bv,
	\eeq 
	where
	\begin{itemize}
		\item $\bs(\tau) \in \bbC^N$ is the sampled delayed signal associated with $s(t)$ and time-delay $\tau$, i.e., $\bs(\tau) =[s(\tau_0-\tau), s(\tau_0+\Delta t-\tau), \cdots, s(\tau_0+(N-1)\Delta t-\tau)]^\trans$, where $\tau_0\approx\tau$ is the first time sample where the echo from the range cell of interest occurs and $N\Delta t=T_\rmp$.
		\item $\ba(f_{\rm d})=[1,e^{\jmath2\pi f_{\rm d}T_\rmr}, \cdots, e^{\jmath2\pi f_{\rm d}(M-1)T_\rmr}]^\trans$ represents the temporal steering vector {(i.e., temporal signature)}.
		\item $\bv\in\bbC^{MN}$ denotes the signal-independent interference vector including additive noise jamming, hot clutter, {terrain-bounce or terrain scattered interference\cite{li2006signal,de2009design}}. It is modeled as a zero-mean complex circular Gaussian vector with positive definite covariance matrix $\mathbb{E}[\bv\bv^\ctrans]=\bSigma_{\rmv }$. Herein, the focus is on the slow-time domain, assuming that the signal-independent interference in the fast-time domain is white. Consequently, the covariance matrix can be cast as  $\bSigma_{\rmv}=\bSigma_{\rmt}\otimes \bI_N$, where $\bSigma_{{\rmt}}\in \bbH_{++}^{M}$ indicates the covariance matrix in slow-time.
	\end{itemize}
	
	\subsection{Performance Measures}
	\subsubsection{Detection Probability}
	The problem of detecting a target using the signal model in \eqref{echo} can be formulated as the following binary hypothesis test
	\beq \label{binhypo}
		\left\{ {\begin{array}{*{20}{l}}
				{H_0}: \br = \bv \\
				{H_1}: \br = \alpha (\ba(f_{\rm d}) \! \odot \! \bc) \! \otimes \! {\bs(\tau)} + \bv
		\end{array}} \right..
	\eeq
	Assuming that the interference covariance matrix $\bSigma_{\rmv}$ is known, the generalized likelihood ratio test (GLRT) detector over $\alpha$ for Problem \eqref{binhypo} is equivalent to the optimum test according to the Neyman-Pearson criterion under the assumption that the phase of $\alpha$ is uniformly distributed in $[0, 2\pi)$, and is given by 
	\beq
		|\br^\ctrans\bSigma_{\rmv}^{-1}((\ba(f_{\rm d}) \! \odot \! \bc) \! \otimes \! {\bs}(\tau))|^2\LRT{H_1}{H_0} \eta,
	\eeq
	where $\eta$ indicates the detection threshold determined by the desired value of the false alarm probability ($P_{\text{fa}}$). The Gaussian assumption for the interference component, implies that the detection probability can be written as
	\beq
		P_\rmd = Q_1\Big(\sqrt{2 \text{SINR}(\bc, \tau, f_\rmd)}, \sqrt{-2\ln P_{\text{fa}}}\Big), \nonumber
	\eeq
	where $Q_1(\cdot, \cdot)$ is the Marcum Q function of order 1, and $\text{SINR}(\bc, \tau, f_\rmd)$ is the output SINR, defined as 
	\begin{align}\label{SINR1}
		\!\!\text{SINR}(\bc, \tau, f_\rmd)&\!=\! |\alpha|^2[(\ba(f_{\rm d}) \!\! \odot \! \bc) \!\! \otimes \!\! {\bs(\tau)}]^\ctrans  \!\bSigma_{\rmv}^{-1}[(\ba(f_{\rm d}) \!\! \odot \! \bc) \!\! \otimes\! \! {\bs(\tau)}]\nonumber \\
		&\!=\!|\alpha|^2\|\bs(\tau)\|^2(\ba(f_{\rm d}) \! \odot\! \! \bc)^\ctrans\!\bSigma_{\rmt}^{-1} (\ba(f_{\rm d}) \! \odot \! \bc).\!\!\!
		\end{align}
	Under the assumption of $\tau_0 \approx \tau$, the term $\|\bs(\tau)\|^2$ in \eqref{SINR1} can be approximated by
	\begin{align}\label{stau2}
		\!\!\|\bs(\tau)\|^2&=\sum\limits_{n=0}^{N-1}|s(t)|^2_{t=\tau_0+n\Delta t-\tau}\nonumber\\
		&\approx \sum\limits_{n=0}^{N-1}|s(t)|^2_{t=n\Delta t} \approx \frac{1}{\Delta t}\int_{0}^{T_\rmp} |s(t)|^2 dt=N.
	\end{align}
	Then, the output SINR can be simplified as\footnote{For any finite sampling time, the reported equality holds true only approximately. However, as long as the sampling interval is short enough, the approximation becomes tighter and tighter.}
	\beq
		\text{SINR}(\bc, \tau, f_\rmd) = |\alpha|^2N\bc^\ctrans\bM_0\bc,
	\eeq
	where $\bM_0=\bSigma_{\rmt}^{-1} \odot (\ba(f_{\rm d})\ba(f_{\rm d})^\ctrans)\in \bbC^M$.

	The expression of $P_\rmd$ indicates that for a fixed $P_{\text{fa}}$, there is a monotonic relationship between SINR and $P_\rmd$. In other words, the maximization of $P_\rmd$ can be achieved optimizing the SINR over the slow-time radar code.
	
	\subsubsection{Delay and Doppler Accuracy}
	consider the estimation of $\bgamma=[\widetilde \balpha^\trans, \tau, f_{\rm d}]^\trans$, where $\widetilde \balpha\!=\!\big[\Re\{\alpha\}, \Im\{\alpha\}\big]^\trans$ is a nuisance parameter. Denoting by $\bmu({\bgamma})\!=\! \alpha (\ba(f_{\rm d})  \!\odot\!  \bc) \!\otimes\!  {\bs(\tau)}$, the Fisher information matrix $\mathbfcal{I}$ for estimating $\bgamma$ is \cite{dogandzic2001cramer}
	\beq
		\mathbfcal{I} = 2\Re\big\{\bPsi^\ctrans({\bgamma})\bSigma_{\rmv}^{-1}\bPsi({\bgamma})\big\},\label{FIM}
	\eeq
	where 
	\beq
	\bPsi(\bgamma) = \frac{\partial \bmu({\bgamma})}{\partial \bgamma^\trans}=\bigg[\frac{\partial \bmu({\bgamma})}{\partial \widetilde \balpha^\trans},\frac{\partial \bmu({\bgamma})}{\partial \tau},\frac{\partial \bmu({\bgamma})}{\partial f_\rmd}\bigg].\label{Dgamma}
	\eeq
	{Based on the Fisher information matrix, a lower bound on the covariance matrix of the estimation error in $\bgamma$ (for any unbiased estimator) can be derived, which is known as the CRB matrix. The following proposition establishes the CRB matrix for delay and Doppler estimation, expressed as a function of the slow-time code.} 
	\begin{proposition}\label{Prop1}
		Assuming $\alpha$ unknown, the CRB matrix for  $\tau$ and $f_\rmd$ is diagonal, and its diagonal elements are given by
		\begin{subequations}
			\!\!\!\!\!\!\begin{align}
				\begin{split}\label{CRBt}
					\!\!\!\![\textbf{CRB}_{\tau f_\rmd}(\bc)]_{11} &\!= \!\frac{3}{2|\alpha|^2N\pi^2 B^2} \displaystyle \frac{1}{\bc^\ctrans \bM_0 \bc},
				\end{split}\\
				\begin{split}\label{CRBf}
					\!\!\!\![\textbf{CRB}_{\tau f_\rmd}(\bc)]_{22} &\!= \!\frac{1}{2|\alpha|^2N} \displaystyle \frac{\bc^\ctrans \bM_0 \bc}{\bc^\ctrans \!\bM_0 \bc\bc^\ctrans \!\bM_2 \bc \!-\!|\bc^\ctrans \bM_1 \bc|^2},\!\!\!\!
				\end{split}
			\end{align}
		\end{subequations}
	where $\bM_1 = \Diag(\bb^*) \bM_0$, $ \bM_2 =\bb\bb^\ctrans \odot\bM_0 $, and $\bb=[0,\jmath2\pi T_\rmr,\cdots,\jmath2\pi (M-1)T_\rmr]^\trans$. 
	\end{proposition}
	\begin{proof}
		See Appendix \ref{proofCRB} of the supplemental material.
	\end{proof}
	
	For multiple parameters estimation, the determinant of the CRB matrix is monotonically related to the volume of its error ellipsoid\cite{van2001detection}. As a result, the CRB matrix determinant, i.e.,  $\det({\textbf{CRB}_{\tau f_\rmd}(\bc)} )$, can be regarded as a valuable performance metric to control in order to improve delay and Doppler estimation accuracies; it is given by 
	\beq\label{detCRB}
	\begin{aligned}
		\det&({\textbf{CRB}_{\tau f_\rmd}(\bc)} ) \\
		&= \frac{3}{4|\alpha|^4N^2 \pi^2 B^2}  \displaystyle \frac{1}{\bc^\ctrans \bM_0 \bc\bc^\ctrans \bM_2 \bc-|\bc^\ctrans \bM_1 \bc|^2}.
	\end{aligned}
	\eeq
	
	\subsection{Problem Formulation}
	While the target is being tracked, a reliable measurement process entailing trustworthy target detection and accurate parameters estimation, namely, delay and Doppler, has the potential to enhance track maintenance performance. In this respect, it is of interest to simultaneously optimize the SINR and CRB over the radar code utilizing the predicted target's delay and Doppler information from the track file. {To comply with the radar power budget, an energy constraint is forced on the radar sequence: $\|\bc\|^2 = 1$. Furthermore, to bestow some desirable properties to the sought code, such as low peak-to-average power ratio (PAPR), low integrated sidelobe level\footnote{In the slow-time domain, the ISL reflects the capability to suppress range ambiguous returns. Generally, a lower ISL indicates better suppression of range ambiguities.} (ISL), and good Doppler resolution, a similarity constraint \cite{de2008code,cui2014mimo} is imposed, i.e., $\|\bc-\bc_0\|^2 \leq \zeta$}, where $\bc_0$ denotes a reference code with unit energy $\|\bc_0\|^2\!=\!1$ {sharing appropriate hallmarks} and $0\leq\zeta\leq 2$ is a real parameter ruling the degree of similarity\footnote{In the case $\|\bc\|^2=\|\bc_0\|^2=1$, the similarity constraint can be recast as $\Re\{\bc_0^\ctrans\bc\}\geq (2-\zeta)/2$\cite{de2009design}. When $\zeta>2$,  it can be regarded as inactive, since the objective functions are invariant to phase rotation and any feasible point $\bc$ satisfying  $\Re\{\bc_0^\ctrans\bc\}<0$ shares the same objective value as $e^{\jmath\pi}\bc$, which is feasible also for $\zeta=2$. Therefore, values of $\zeta>2$ are not of interest from an optimization standpoint.}. Incorporating the above design guidelines, {for a given delay-Doppler shift pair $(\tau,f_d)$ of interest}, the joint optimization of SINR and CRB can be formulated as a multi-objective optimization problem, expressed as follows:
	\beq \label{Pareto}	
	\calP\left\{\begin{array}{lll}
		\max \limits_{\bc} & (\det^{-1}({\textbf{CRB}_{\tau f_\rmd}(\bc)} ), \text{SINR}(\bc, \tau, f_\rmd))\\
		\subto & \|\bc\|^2 = 1\\
		& \|\bc-\bc_0\|^2 \leq \zeta
	\end{array}
	\right.\!\!\!\!.
	\eeq
	
	$\calP$ is a non-convex problem\footnote{A non-convex multi-objective optimization problem\cite{pardalos2017non} is any problem where at least one of the objective functions or the feasible region is non-convex.} due to the non-concavity of $\det^{-1}({\textbf{CRB}_{\tau f_\rmd}(\bc)} )$ and $\text{SINR}(\bc, \tau, f_\rmd)$ as well as the energy constraint. In the following section, an iterative method based on scalarization technique\cite{de2011pareto, pardalos2017non} and tensor-based relaxation\cite{aubry2013ambiguity} is developed to come up with a good quality solution to $\calP$ in a polynomial-time.
	
	\section{Radar Code Synthesis}\label{SecIII}
	Let us first use the scalarization technique to systematically transform the multi-objective optimization problem $\calP$ into a single-objective optimization problem. Specifically, the scalarized version for $\calP$ is
	\beq \label{Pareto1}
	\calP_\beta\left\{\begin{array}{lll}
		\max \limits_{\bc} & f(\bc)\\
		\subto & \|\bc\|^2 = 1\\
		& \|\bc-\bc_0\|^2 \leq \zeta
	\end{array}
	\right.,
	\eeq
	where 
	\beq\label{fc}
		f(\bc)\!=(1\!-\!\beta)(\bc^\ctrans \bM_0 \bc\bc^\ctrans \bM_2 \bc\!-\!|\bc^\ctrans \bM_1 \bc|^2)\!+\!\beta \bc^\ctrans \bM_0 \bc
	\eeq
	with $\beta \in[0, 1]$ being the scalarization weight. A worthwhile remark on the scalarization technique is now provided.
	\begin{remark}
		The parameter $\beta$ {controls} a trade-off between SINR and CRB, delineating the cost required to enhance one figure of merit at the expense of the other. Precisely, when $\beta=0$ in $\calP_\beta$, the focus is solely on CRB optimization. As $\beta$ increases, the emphasis gradually shifts towards SINR optimization, culminating at $\beta = 1$, where the focus entirely entails SINR optimization. 
	\end{remark}
	
	\begin{remark}
		Notice that $\calP_\beta$ remains, in general, a NP-hard problem\footnote{For a compact set, like that defined by the energy and similarity constraints in our study, quartic polynomial optimization yields NP-hard problems, in general\cite{he2014inhomogeneous}.}. Accordingly, its polynomial-time synthesized solution is an approximation of the Pareto point to $\calP$ associated with the chosen $\beta$ parameter. 
	\end{remark}
	
	In order to tackle the non-convex problem $\calP_\beta$, a tensor-based relaxation method is developed \cite{aubry2013ambiguity}. In this respect, it is first necessary to recast the objective function into an equivalent way adding two appropriate terms $\mu_1(\bc^\ctrans\bc)^2$, and $\mu_2(\bc^\ctrans\bc)$, where $\mu_1\geq 0$ and $\mu_2\geq 0$ are bespoke finite constants. This is just a technical trick that comes in handy to develop a solution technique leading to the synthesis of a high quality code. Hence, the objective function (owing to the forced constraints) can be equivalently recast as
	{\beq\label{fc1}
		\widetilde f(\bc) \!=f(\bc)\!+\!  \mu_1(\bc^\ctrans\bc)^2 \!+\! \mu_2(\bc^\ctrans\bc) ,
	\eeq}
	 with the resulting {equivalent} optimization problem
	\beq \label{Pareto2}
	{\overline \calP}_\beta\left\{\begin{array}{lll}
		\max\limits_{\bc} & \widetilde f(\bc) \\
		\subto & \|\bc\|^2 = 1\\
		& \|\bc-\bc_0\|^2 \leq \zeta
	\end{array}
	\right..
	\eeq
	
	To proceed further and handle ${\overline \calP}_\beta$ by means of the tensor-based relaxation method in \cite{aubry2013ambiguity}, let us represent the complex inhomogeneous quartic polynomial\footnote{The inhomogeneous quartic polynomial $\widetilde f(\bc)$ with respect to $\bc$ can be homogenized into a quartic function of $\widetilde \bc=[\bc^\trans, 1]^\trans$.} $\widetilde f(\bc)$, via a fourth order conjugate super-symmetric tensor \cite{sidiropoulos2017tensor}, i.e.,
	\begin{align}\label{Tensor_F}
		\widetilde f(\bc)=&\calF\bigg(\binom{\widetilde \bc}{\widetilde \bc^*}, \binom{\widetilde \bc}{\widetilde \bc^*}, \binom{\widetilde \bc}{\widetilde \bc^*}, \binom{\widetilde \bc}{\widetilde \bc^*}\bigg)
	\end{align}
	where $\widetilde \bc=[\bc^\trans, 1]^\trans\in\bbC^{M+1}$ and the tensor  $\calF\!\in\!\bbC^{(2M+\!2)\times(2M+\!2)\times(2M+\!2)\times(2M+\!2)}$ is defined in Equation \eqref{Fpi} (see Appendix \ref{proof_tensor} of the supplemental material). Consequently, ${\overline \calP}_\beta$ can be equivalently recast as 
	\beq \label{P_detCRB2_1}
	\widetilde \calP_\beta\left\{\begin{array}{lll}
		\!\!\!\max\limits_{\bc^1, \bc^2, \bc^3, \bc^4} & \calF\bigg(\disps\binom{\widetilde \bc^1}{\widetilde \bc^{1*}}, \binom{\widetilde \bc^2}{\widetilde \bc^{2*}}, \binom{\widetilde \bc^3}{\widetilde \bc^{3*}}, \binom{\widetilde \bc^4}{\widetilde \bc^{4*}}\bigg)\\[1.2em]
		\quad \; \subto & \|\bc^p\|^2 = 1\\
		& \|\bc^p-\bc_0\|^2 \leq \zeta, p=1,2,3,4\\
		& \bc^1=\bc^2=\bc^3=\bc^4
	\end{array}
	\right.\!\!\!\!\!,
	\eeq
	where $\widetilde \bc^p=[(\bc^p)^{\trans}, 1]^\trans,p=1,2,3,4$.
	
	The relaxation technique procedure to solve $\widetilde \calP_\beta$ is reported in Algorithm \ref{algo1}. Before proceeding further, it is worth mentioning the following key result whose proof is provided in Appendix C in \cite{aubry2013ambiguity}.
	\begin{theorem*}\label{theo1} 
		Suppose $\widetilde f(\bx)$ is a convex complex quartic function and let $\calF$ be the conjugate super-symmetric tensor form associated with $\widetilde f(\bx)$, then
		\begin{align}
		\!\!\!&\max\!\Big\{ \widetilde f(\bx^1), \widetilde f(\bx^2), \widetilde f(\bx^3), \widetilde f(\bx^4)\Big\}\nonumber  \\
		&\quad\quad\quad \geq \calF\bigg(\disps\binom{ \bx^1}{ \bx^{1*}}, \binom{ \bx^2}{ \bx^{2*}}, \binom{ \bx^3}{ \bx^{3*}}, \binom{ \bx^4}{ \bx^{4*}}\bigg).
		\end{align}
	\end{theorem*} 
	
	According to Theorem \ref{theo1}, if $\mu_1$ and $\mu_2$ can be determined such that $\widetilde f(\bc)\! =\! \mu_1(\bc^\ctrans\bc)^2 +\!\mu_2(\bc^\ctrans\bc)+\! f(\bc)$ is convex, then
	\begin{align}
	\widetilde f(\bc^\star) \geq \calF\bigg(\disps\binom{\widetilde \bc^{1,\star}}{(\widetilde \bc^{1,\star})^*}, \binom{\widetilde \bc^{2,\star}}{(\widetilde \bc^{2,\star})^*}, \binom{\widetilde \bc^{3,\star}}{(\widetilde \bc^{3,\star})^*}, \binom{\widetilde \bc^{4,\star}}{(\widetilde \bc^{4,\star})^*}\bigg),
	\end{align}
	where $\bc^\star\!=\!\arg \displaystyle \max\limits_{\bc^{p,\!\star}, \, p=1,2,3,4} \{\widetilde f(\bc^{p,\star})\}$. This key property ensures that a solution to $\widetilde \calP_\beta$ ensuring an objective value greater than or equal to the optimized value of $\widetilde \calP_{\beta}^\text{relax}$ can be obtained through the last step in Algorithm \ref{algo1}.
	
	According to these guidelines, in the following subsections, the focus is on developing an effective optimization algorithm capable of yielding  a good quality solution to $\widetilde \calP_{\beta}^\text{relax}$, as well as finding suitable $\mu_1$ and $\mu_2$ that endow convexity to $\widetilde f(\bc)$.

	\begin{algorithm}[t]
		\setstretch{1.15}
		\caption{ Relaxation method for solving $\widetilde \calP_\beta$.}\label{algo1} 
		\begin{algorithmic}[1]
			\State Drop the constraint $\bc^1=\bc^2=\bc^3=\bc^4$ in $\widetilde \calP_\beta$;
			\State Solve the following relaxed problem
			\beq \label{P_detCRB3}
			\!\!\!\widetilde \calP_{\beta}^\text{relax}\left\{\begin{array}{lll}
				\!\!\!\max\limits_{\bc^1, \bc^2, \bc^3, \bc^4} & \calF\bigg(\disps\binom{\widetilde \bc^1}{\widetilde \bc^{1*}}, \binom{\widetilde \bc^2}{\widetilde \bc^{2*}}, \binom{\widetilde \bc^3}{\widetilde \bc^{3*}}, \binom{\widetilde \bc^4}{\widetilde \bc^{4*}}\bigg)\\[1.2em]
				\quad \; \subto & \|\bc^p\|^2 = 1\\
				& \|\bc^p-\bc_0\|^2 \leq \zeta, p=1,2,3,4
			\end{array}
			\right.\!\!\!\!\!,
			\eeq
			and denote by $\bc^{1,\star}, \bc^{2,\star}, \bc^{3,\star}, \bc^{4,\star}$ the optimized solutions;\label{step2}%
			\State\label{step3} Construct an optimized solution to $\widetilde \calP_\beta$ as 
			\beq
				\bc^\star\!=\!\arg \displaystyle \max\limits_{\bc^{p,\!\star}, \, p=1,2,3,4} \{\widetilde f(\bc^{p,\star})\}.
			\eeq
		\end{algorithmic}
	\end{algorithm}
	
	\subsection{MBI for Solving \rm{$\widetilde \calP_{\beta}^\text{relax}$}}
	{Note that the constraints on the blocks $\bc^1, \bc^2, \bc^3, \bc^4$ are decoupled, making alternating optimization a viable approach for solving $\widetilde \calP_{\beta}^\text{relax}$. However, the corresponding updating rule fails to guarantee convergence to a stationary point, in general. To overcome this shortcoming, in this subsection, an iterative method based on the MBI framework is developed to solve $\widetilde \calP_{\beta}^\text{relax}$.} Specifically, at each iteration, every single block in $\{\bc^1, \bc^2, \bc^3, \bc^4\}$ is optimized while keeping the others fixed, and only the block yielding the maximum increase of the objective function is updated. Aimed at reducing the computational burden associated with tensor operations, let us introduce an equivalent expression of the objective function in  $\widetilde \calP_{\beta}^\text{relax}$. As shown in Appendix \ref{proof_tensor} of the supplemental material, 
	\begin{align}\label{f_cccc}
		&\calF\bigg(\disps\binom{\widetilde \bc^1}{\widetilde \bc^{1*}}, \binom{\widetilde \bc^2}{\widetilde \bc^{2*}}, \binom{\widetilde \bc^3}{\widetilde \bc^{3*}}, \binom{\widetilde \bc^4}{\widetilde \bc^{4*}}\bigg)\nonumber\\
		&\quad= \frac{1}{24} \!\!\sum\limits_{(p_1, p_2, p_3, p_4) \atop \in \Pi_4 (1, 2, 3, 4)}\!\!\!\! \Big\{\mu_1 (\bc^{p_1})^\ctrans\bc^{p_2}(\bc^{p_3})^\ctrans\bc^{p_4}\nonumber\\
		&\quad\quad\quad\quad\quad\quad +\!(1-\beta)(\bc^{p_1})^\ctrans\bM_0\bc^{p_2}(\bc^{p_3})^\ctrans\bM_2\bc^{p_4}\nonumber\\
		&\quad\quad\quad\quad\quad\quad-\!(1-\beta)(\bc^{p_1})^\ctrans\bM_1\bc^{p_2}(\bc^{p_3})^\ctrans\bM_1^\ctrans\bc^{p_4}\Big\}\nonumber\\
		&\quad\quad+\frac{1}{12}\!\!\sum\limits_{(p_1, p_2) \atop \in \Pi_2 (1, 2, 3, 4)}\!\!\!\!\!\!\!\Big\{\mu_2(\bc^{p_1})^\ctrans\bc^{p_2}\!+\!\beta(\bc^{p_1})^\ctrans\bM_0\bc^{p_2}\Big\}\nonumber\\
		&\quad=\widetilde f_{_\text{ML}}(\bc^1, \bc^2, \bc^3, \bc^4).
	\end{align}
	Needless to say, $\widetilde f_{_\text{ML}}(\bc^1, \bc^2, \bc^3, \bc^4) \!\!=\!\! \widetilde f(\bc)$, if $\bc^1\!\!=\!\!\bc^2\!\!=\!\!\bc^3\!\!=\!\!\bc^4\!\!=\!\!\bc$.
	
	In line with MBI-based optimization process, let $\bc^p_{(n)}, p=1,2,3,4$ be the $p$-th block of at the $n$-th iteration; then, the problem for the update of $p$-th block at the $(n\!+\!1)$-th iteration is given by
		\beq \label{P_cp} 
		\begin{array}{lll}
			\disps \max\limits_{\bc^p_{}} & \widetilde f_{_\text{ML}}(\bc^p; \{\bc^{\tilde p}_{(n)}\}_{\tilde p \in \calS_p})\\
			\subto & \|\bc^p\|^2 = 1\\
			& \|\bc^p-\bc_0\|^2 \leq \zeta
		\end{array},
		\eeq
		where $\calS_p=\{1,2,3,4\}\setminus \{p\}$, and 
		\beq
			\widetilde f_{_\text{ML}}(\bc^p; \{\bc^{\tilde p}_{(n)}\}_{\tilde p \in \calS_p})= \Re\big\{(\bd_{(n)}^p)^\ctrans\bc^p\big\}+d_{(n)}^p
		\eeq
		is the restriction of the function $\widetilde f_{_\text{ML}}(\bc^1, \bc^2, \bc^3, \bc^4)$ to $\bc^p$, with the other blocks fixed at the values of the previous iteration,  $\bd_{(n)}^p$ and $d_{(n)}^p$ specified in \eqref{dnp} as detailed in Appendix \ref{proofP_inMBI} of the supplemental material. Utilizing the unit energy $\|\bc^p\|^2=1$ to recast the similarity constraint as $2\Re\{\bc_0^\ctrans\bc^p\}\geq 2- \zeta$ \cite{de2009design}, the equivalent optimization problem for Problem \eqref{P_cp} is 
		\beq \label{P_cp1}
		\begin{array}{lll}
			\disps \max\limits_{\bc^p_{}} &  \widetilde f_{_\text{ML}}(\bc^p; \{\bc^{\tilde p}_{(n)}\}_{\tilde p \in \calS_p})\\
			\subto & \|\bc^p\|^2 = 1\\
			& 2\Re\{\bc_0^\ctrans\bc^p\}\geq 2- \zeta
		\end{array}.
		\eeq
		
		Problem \eqref{P_cp1} is hidden-convex\cite[Section III.B]{aubry2013knowledge}, and its optimal solution can be obtained leveraging the following lemma.
		
		\begin{lemma}\label{lemma3_1}
			 The optimal solution to Problem \eqref{P_cp1} can be constructed from the optimal solution to Problem \eqref{P_cp2}.
				\beq \label{P_cp2}
				\begin{array}{lll}
					\disps \max\limits_{\bc^p_{}} & \widetilde f_{_\text{ML}}(\bc^p; \{\bc^{\tilde p}_{(n)}\}_{\tilde p \in \calS_p})\\
					\subto & \|\bc^p\|^2 \leq 1\\
					& 2\Re\{\bc_0^\ctrans\bc^p\}\geq 2- \zeta
				\end{array}.
				\eeq
		\end{lemma}
		\begin{proof}
			See Appendix \ref{prooflemma3_1} of the supplemental material.
		\end{proof}
		
		Since Problem \eqref{P_cp2} is convex and the constraints satisfy the Slater’s condition\cite{boyd2004convex}, the strong duality holds. Hence, an optimal solution, i.e., a solution that fulfills the energy equality can be obtained via the Karush–Kuhn–Tucker (KKT) optimality conditions, is given by
		\begin{align}\label{bcp_opt}
		\bc^{p} \!\!= \!\!\disps \left\{\!\!\!\begin{array}{lll}
			\bc_0, &\!\!\!\! \text{if} \; \bd_{(n)}^p\!=\!\bzero_M, \\
			\!\!\delta \bc_0 \!+\! \sqrt{1\!-\!\delta^2}\disps \frac{\bc_{0}^{\perp}}{\|\bc_0^{\perp}\|}\!, &\!\!\!\! \text{if} \; \bd_{(n)}^p\!=\!-\varrho\bc_0, \varrho\!>\!0, \\
			\displaystyle \frac{\bd_{(n)}^p}{\|\bd_{(n)}^p\|}, &\!\!\!\! \text{if} \; {2\Re \!\big\{ {{\bc}_0^\ctrans\bd_{(n)}^p} \!\big\}\!+\! \|\bd_{(n)}^p\|(\zeta\!-\!2)\! \geq \! 0}, \\[1.2em]
			\disps \frac{{\bd_{(n)}^p\!+\!2\lambda_2\bc_0}}{\|\bd_{(n)}^p\!+\!2\lambda_2\bc_0\|}, &\!\!\!\! \text{otherwise},
		\end{array}
		\right.\!\!\!\!\!\!
		\end{align}
		where $\delta=(2-\zeta)/2$, and $\lambda_2$ specified in \eqref{pos_lam2}  of Appendix \ref{proofsolutionP_cp} of the supplemental material.
	
	\begin{algorithm}[t]
		\setstretch{1.15}
		\caption{ MBI for solving $\widetilde \calP_{\beta}^\text{relax}$.}\label{algo2}
		\algorithmicrequire{ Reference code $\bc_0$, similar parameter $\zeta$, initial feasible code $\bc_{(0)}$, the minimum required improvement $\epsilon$, and the maximum iteration number $N_{\rm iter}$ {to facilitate possible real-time implementation.}} \\
		\algorithmicensure{ Optimized solutions $\bc^{1, \star},\bc^{2, \star}, \bc^{3, \star}, \bc^{4, \star}$.} 
		\begin{algorithmic}[1]
			\State \multiline{(Initialization): \\
				Set $n=0$; Let $\bc_{(0)}^p=\bc_{(0)}, p=1,2,3,4$, and $\upsilon_{(0)} = \widetilde f_{_\text{ML}}(\bc_{(0)}^1, \bc_{(0)}^2,\bc_{(0)}^3,\bc_{(0)}^4)$;}
			\While { $|\upsilon_{(n+1)}-\upsilon_{(n)}|$ $\ge \epsilon$ or $n\leq N_{\rm iter}$ }
			\State \multiline{(Block Improvement): For each $p\!=\!1,2,3,4$, obtain the solution to Problem \eqref{P_cp1} via \eqref{bcp_opt}, and let $\widehat \bc_{(n+1)}^p$ be the solution, $\upsilon_{(n+1)}^p \!\!=\!\! \widetilde f_{_\text{ML}}(\bc_{(n)}^1,\cdots, \widehat \bc_{(n+1)}^p,\cdots\!,\bc_{(n)}^4)$ be the corresponding optimized objective value;}
			\State \multiline{(Maximum Improvement): Let $  \hat p\!=\!\arg \max\limits_{1\le p\le4}\! \{\upsilon_{(n+1)}^p\}$, and update
				\vspace{-4mm}
				\begin{subequations}
					\begin{eqnarray}
						\upsilon_{(n+1)} &\!\!\!= \!\!\!&\upsilon_{(n+1)}^{\hat p},\\
						\bc_{(n+1)}^{\hat p} &\!\!\!=\!\!\!& \widehat \bc_{(n+1)}^{\hat p}, \\
						\bc_{(n+1)}^p &\!\!\!= \!\!\!& \bc_{(n)}^p, \forall p \ne \hat p;
					\end{eqnarray}	
				\end{subequations}\vspace{-4mm}}
			\State Set $n = n+1$;
			\EndWhile 
			\State (Output): $\bc^{p, \star} \!\!= \!\bc^{p}_{(n+1)}, p\!=\!\!1,2,3,4$.
		\end{algorithmic}
	\end{algorithm}
	
	The iterative procedure for solving $\widetilde \calP_{\beta}^\text{relax}$ is summarized in Algorithm \ref{algo2}. As to the computational complexity of Algorithm \ref{algo2}, it is linear with the number of iterations. Specifically, in each iteration, it requires constructing $\bd_{(n)}^p, d_{(n)}^p, p=1,2,3,4$ in \eqref{dnp}, which involves multiple matrix multiplications with a dominating computational complexity of $\calO(M^2)$. {Remarkably, the iterative procedure based on the MBI framework can be triggered using the optimized solution obtained from the plain alternating rule, and executed in parallel, thereby achieving a more favorable performance and computational complexity\cite[Section III]{aubry2018new}.}
	
	The convergence of Algorithm \ref{algo2} is now discussed. At the $n$-th iteration, denote the multi-linear form objective value as 
	\beq\label{vn}
		\upsilon_{(n)} \!=\! \widetilde f_{_\text{ML}}(\bc_{(n)}^1,\bc_{(n)}^2, \bc_{(n)}^3,\bc_{(n)}^4).
	\eeq
	According to \cite[Theorem 2]{razaviyayn2013unified}, the sequence $\upsilon_{(n)}$ is monotonically increasing and converges, and any cluster point of the sequences of solutions produced by Algorithm \ref{algo2} is a stationary point to $\widetilde \calP_{\beta}^\text{relax}$. 
	
	\subsection{Method to Find Suitable $\mu_1$ and $\mu_2$}
	In this subsection, guidelines for a bespoke selection of the parameters $\mu_1$ and $\mu_2$ are provided. A key {role} is played by the following proposition.
	\begin{proposition}\label{Propo2}
		 Let $\bD \!\in\! \bbH^{M} $, and either $\bA_j, \bB_j \!\in\! \bbH^{M}$ or $\bA_j\!\!=\!\! \bB_j\!\!\in\!\!\bbC^{M\times M},j\!\!=\!\!1,2,\cdots,J$. For a real-valued complex inhomogeneous quartic function $g_{_\rmI}(\bx) \!=\! \alpha_0\bx^\ctrans \bD\bx\!+\!\sum\limits_{j=1}^{J}\alpha_j\bx^\ctrans\bA_j\bx\bx^\ctrans\bB_j^\ctrans\bx$ with $\alpha_0,\alpha_{j}\!\!\in\!\bbR$, $\widetilde g(\bx)\!=\!\mu_1(\bx^\ctrans\bx)^2\!+\!\mu_2(\bx^\ctrans\bx)\!+\!g_{_\rmI}(\bx)$ is convex, if the following condition holds
		\begin{subequations}
		\begin{align}
			\mu_1 &\!\geq\; \displaystyle \max\limits_{\|\by\|^2=1, \|\bz\|^2=1} -\frac{1}{2}\sum\limits_{j=1}^{J}{\alpha_j} h(\by, \bz, \bA_j, \bB_j), \\
			\mu_2 &\!\geq\; \max \{0, -\lambda_{\min}(\alpha_0\bD)\},
		\end{align}
		\end{subequations}
		where 
		\beq \label{h_yz}
		\begin{aligned}
			\!\!h(\by, \bz, \bA_j, \bB_j) =  &\by^\ctrans\! \bA_j \by \bz^\ctrans\! \bB_j^\ctrans \bz + \by^\ctrans\! \bA_j \bz \by^\ctrans \!\bB_j^\ctrans \bz \\
			&+ \by^\ctrans\! \bA_j \bz \bz^\ctrans\! \bB_j^\ctrans \by+ \bz^\ctrans\! \bA_j \by \by^\ctrans\! \bB_j^\ctrans \bz \\
			&+ \bz^\ctrans\! \bA_j \by \bz^\ctrans\! \bB_j^\ctrans \by + \bz^\ctrans\! \bA_j \bz \by^\ctrans\! \bB_j^\ctrans \by, \\
			&\quad\quad\quad\quad\quad\quad\quad\quad\quad\quad\forall \by,\bz\in\bbC^M. \!\!\!\!
		\end{aligned}
		\eeq 
	\end{proposition}
	\begin{proof}
	See Appendix \ref{proof_propo2} of the supplemental material.
	\end{proof}
	Leveraging Proposition \ref{Propo2} with $J=2$, $\bD=\beta\bM_0\succeq 0$, $\alpha_0=1$, $\alpha_1=1-\beta, \alpha_2=-1+\beta$, $\bA_1=\bM_0,\bB_1=\bM_2$, and $\bA_2=\bB_2=\bM_1$, $\widetilde f(\bc)$ is convex provided that $\mu_2\geq0$, and $\mu_1$ is greater than or equal to the optimal value of
	\beq \label{P_findmu}
	\left\{\begin{array}{lll}
		\max\limits_{\by, \bz} & \frac{1}{2}(1-\beta)\widetilde h(\by, \bz)\\
		\subto & \|\by\|^2 = 1\\
		& \|\bz\|^2 = 1
	\end{array}
	\right.,
	\eeq
	where $\widetilde h(\by, \bz)=h(\by, \bz, \bM_1, \bM_1)-h(\by, \bz, \bM_0, \bM_2)$. Problem \eqref{P_findmu} is in general NP-hard, making it difficult to find the global maximum of the objective function. To ensure that the value of $\mu_1$ guarantees the desired property, i.e., the convexity of $\mu_1(\bc^\ctrans\bc)^2 \!+\! \mu_2(\bc^\ctrans\bc) \!+\! f(\bc)$, it is enough to determine the maximum value of an upper bound to the function $\widetilde h(\by, \bz)$. {Leveraging the majorization framework\cite{sun2016majorization}}, an upper-bound function for $\widetilde h(\by, \bz)$ is given by
	\beq\label{h_up}
		\widetilde h_{\text{up}}(\by, \bz) = \lambda_{\max}(\bPhi) (\|\by\|^2+\|\bz\|^2)^2,
	\eeq
	where $\bPhi$ is provided in \eqref{Phi} of Appendix \ref{proofupbound} of the supplemental material. The corresponding optimization problem is
	\beq \label{p_bound}
	\left\{\begin{array}{lll}
		\max\limits_{\by, \bz} & \frac{1}{2}(1-\beta)\lambda_{\max}(\bPhi) (\|\by\|^2+\|\bz\|^2)^2\\
		\subto & \|\by\|^2 = 1\\
		& \|\bz\|^2 = 1
	\end{array}
	\right..
	\eeq
	
	The optimal value for \eqref{p_bound} is $2(1-\beta)\lambda_{\max}(\bPhi)$, thus, $\mu_1$ can be chosen as this same value.
	\section{Numerical Results}\label{SecIV}
	Numerical experiments are conducted to evaluate the performance of the proposed method for radar code design. 
	The transmitted pulse train exploits $M =$ 32 pulses with PRI $T = $ 250 ${\upmu}$s. The bandwidth $B$ and pulsewidth $T_\rmp$ of the chirp signal are set to 5 MHz and 10 $ \upmu$s, respectively. The sampling time is $\Delta=1/2B$ leading to $N=100$. The signal-independent interference covariance matrix in the slow-time domain ${\bSigma_{\sss \rmS}}$ is modeled as\cite{de2008code, de2009design}. 
	\beq
	{\bSigma_{\sss \rmS}}(i, j) = \rho^{|i-j|},
	\eeq 
	where the one-lag correlation coefficient is $\rho = 0.8$. {To guide radar code design, two similarity sequences with constant modulus and low ISL are considered: P3 \cite[pp. 118-122]{levanon2004radar} and generalized Barker code \cite[pp. 109-113]{levanon2004radar}.} Unless otherwise stated, the P3 code serves as the default reference code at the design stage due to its linear phase shift. The received power and normalized Doppler frequency of the target are set to $|\alpha|^2=-$20 dB and $f_\rmd T_\rmr=$ 0.15, and $P_{\text {fa}}=10^{-6}$. The coefficients of the constant term are constructed as $\mu_1=2(1-\beta)\lambda_{\max}(\bPhi)$, $\mu_2=$ 0. Algorithm \ref{algo2} is {initialized} by the reference code $\bc_0$ with stopping conditions $\epsilon=10^{-7}$ and $N_{\rm iter}=$ 6000. 
	
	
	As benchmark codes, let us consider the sequences which maximize $\det^{-1}({\textbf{CRB}_{\tau f_\rmd}(\bc)} )$ or $\text{SINR}(\bc, \tau, f_\rmd)$ under energy constraint only, that is, 
	\begin{align}
		\bc_{\sss \text{BM}}^{\text{S}} &= \arg \displaystyle \max\limits_{\|\bc\|^2=1} \;\bc^\ctrans\bM_0\bc, \\[0.2em]
		\bc_{\sss \text{BM}}^{\text{C}} &= \arg \displaystyle \max\limits_{\|\bc\|^2=1} \;\{\bc^\ctrans \bM_0 \bc\bc^\ctrans \bM_2 \bc-|\bc^\ctrans \bM_1 \bc|^2\}.
	\end{align}
	According to the Rayleigh quotient, $\bc_{\sss \text{BM}}^{\text{S}}$ is the normalized eigenvector corresponding to the maximum eigenvalue of the matrix $\bM_0$. A high quality $\bc_{\sss \text{BM}}^{\text{C}}$ can be found using the proposed algorithm letting $\beta=$ 0, and $\zeta= 2$ (the similarity constraint vanishes). Therefore, the benchmark $P_\rmd$, CRB for Doppler and delay estimation, and CRB determinant are, respectively, defined as
	\beq
	\begin{aligned}
		P_\rmd^{\text{BM}} &= Q\Big(\sqrt{2 |\alpha|^2N(\bc_{\sss \text{BM}}^{\text{S}})^\ctrans\bM_0\bc_{\sss \text{BM}}^{\text{S}}}, \sqrt{-2\ln P_{\text{fa}}}\Big),\\[0.2em]
		\text{CRB}_{\tau}^{\text{BM}} &= [\textbf{CRB}_{\tau f_\rmd}(\bc_{\sss \text{BM}}^{\text{C}})]_{11}, \\[0.2em]
		\text{CRB}_{f}^{\text{BM}} &= [\textbf{CRB}_{\tau f_\rmd}(\bc_{\sss \text{BM}}^{\text{C}})]_{22},\\[0.2em]
		\text{detCRB}^{\text{BM}} &=\text{CRB}_\tau^{\text{BM}}\cdot \text{CRB}_{f}^{\text{BM}}.		
	\end{aligned}\nonumber
	\eeq
	
	In order to compare the performance of our algorithm with that of the similarity code (generalized Barker code), the detection probability, CRB for Doppler and delay estimation, and CRB determinant corresponding to $\bc_0$ are denoted as $P_\rmd^{0}$, $\text{CRB}_f^{0}$,  $\text{CRB}_\tau^{0}$, and $\text{detCRB}^{0}$, respectively.

	
	
	\subsection{Convergence of the Proposed Relaxed Method}\label{sub_convergence}

	Figure \ref{f1} depicts the relaxed multi-linear objective value $\upsilon_{(n)}$ achieved by the MBI method versus the iteration number for $\zeta=$ 0.1, 0.4, 1, and $\beta=$ 0, 0.001, 0.01. The corresponding converged relaxed objective value ($\upsilon_{\text{relax}}^\star$), and the output objective value ($\upsilon_o^\star$) for the solution constructed in the last step of {Algorithm }\ref{algo1} are reported in Table \ref{table1}. As expected, $\upsilon_{(n)}$ monotonically increases as the number of iterations increases, and the output objective  $\upsilon_o^\star$ is greater than or equal to the converged relaxed objective value $\upsilon_{\text{relax}}^\star$. This reveals that, under our selection for $\mu_1, \mu_2$, the relaxation method is capable of providing radar codes with better and better performance along the iterations. Moreover, regardless of $\beta$, the larger the similarity parameter $\zeta$ the higher the achieved objective values due to the enlarged feasible set of $\widetilde \calP_{\beta}^\text{relax}$. 
	\begin{figure}[htbp]
		\centering
		\includegraphics[width=3.3in]{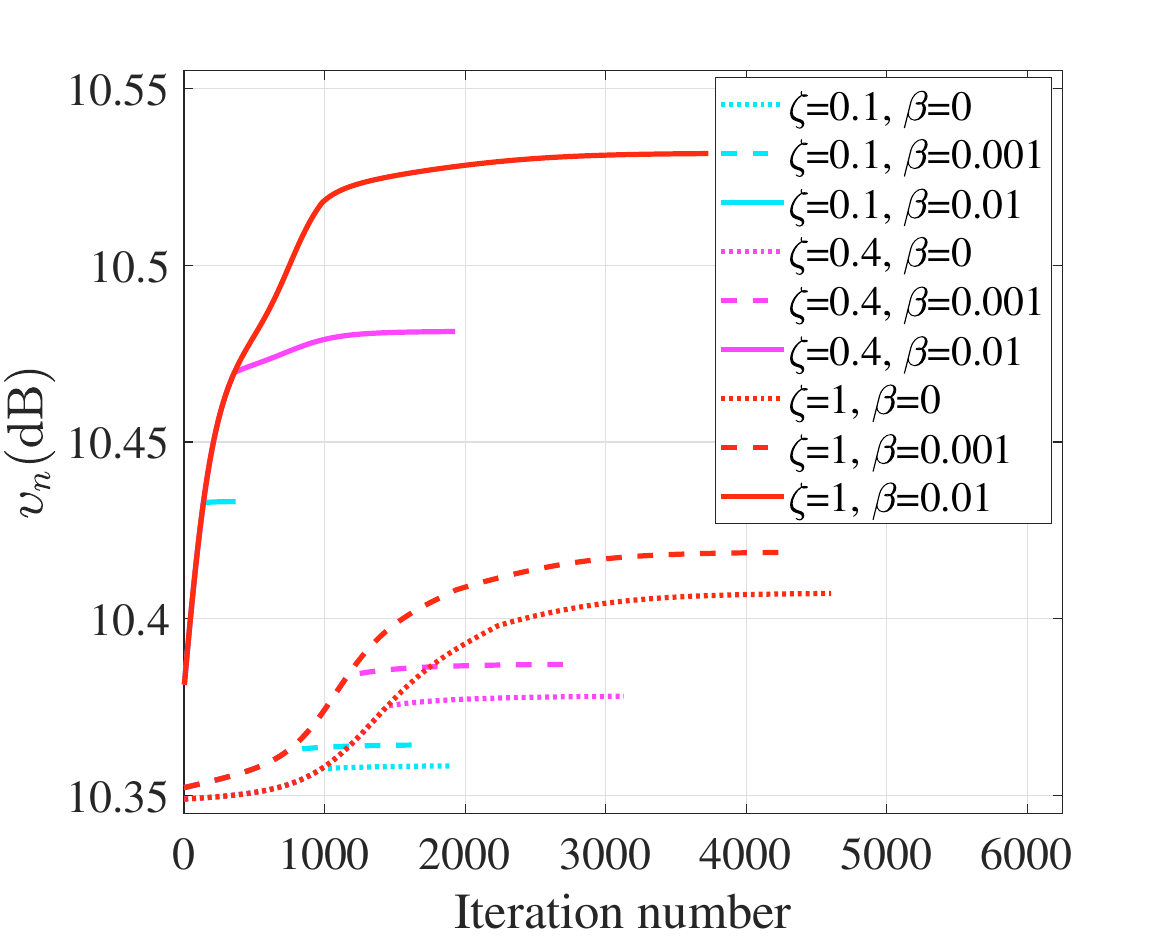}
		\caption{ \small Objective value $\upsilon_{n}$ versus iteration number for $\zeta=$ 0.1, 0.4, 1, and $\beta=$ 0, 0.001, 0.01. }\label{f1}
	\end{figure}

	\begin{table}[htbp]
		\centering  
		\renewcommand{\arraystretch}{1.2}
		\caption{\small Achieved objective values (in dB) for $\zeta=$ 0.1, 0.4, 1, and $\beta=$ 0, 0.001, 0.01.}\label{table1}  
		\setlength{\tabcolsep}{1mm}
		\begin{tabular}{|c|| c|c|c|c|c|c|}
			\hline  
			\multirow{2}{*}{$\beta$} & \multicolumn{2}{c|}{$\zeta=$ 0.1} & \multicolumn{2}{c|}{$\zeta=$ 0.4} & \multicolumn{2}{c|}{$\zeta=$ 1} \\
			\cline{2-7}  
			& $\upsilon_{\text{relax}}^\star$ &  $\upsilon_o^\star$ & $\upsilon_{\text{relax}}^\star$&  $\upsilon_o^\star$ & $\upsilon_{\text{relax}}^\star$ &  $\upsilon_o^\star$  \\
			\hline    
			\hline 
			0 & 10.358    & 10.358 & 10.378 &10.378& 10.407 & 10.407 \\ 
			\hline
			0.001 & 10.364 & 10.364& 10.387&10.387&10.419&10.419\\ 
			\hline
			0.01 & 10.433 &  10.433 &10.48&10.48 &10.532&10.532\\
			\hline
		\end{tabular}
	\end{table}

	\subsection{Detection and Estimation Performance}\label{sub_performance}

	
	Changing $\beta$ in the interval $[0, 1]$ and denoting by $\bc^\star_\beta$ the optimized solution to $\widetilde \calP_\beta$, the curve
	\begin{align}
		\left\{\!\!\!\!\!\!\!\!\!\begin{array}{lll}
			&\!\!\!\!\!\!{\text {SINR}}(\beta) \!\!=\!\!\hat \alpha_1 {\bc_\beta^{\star\ctrans}}\! \bM_{\!0} \bc_\beta^\star \\
			\epc \\
			&\!\!\!\!\!\!\det({\textbf{CRB}})^{-1}(\beta) =\!\!\displaystyle \hat \alpha_2 \big(\bc_\beta^{\star\ctrans} \!\bM_{\!0} {\bc_\beta^\star}\bc_\beta^{\star\ctrans} \!\bM_{\!2} \bc_\beta^\star
			 -|\bc_\beta^{\star \ctrans} \!\bM_{\!1} \!{\bc_\beta^\star}|^2\big)
		\end{array}
		\right.\!\!\!\!\!,\nonumber
	\end{align}
	with $\hat \alpha_1=|\alpha|^2N, \hat \alpha_2 = {\!4|\alpha|^4\!N^2 \pi^2\! B^2\!\!}/{3}$, is referred to as the Pareto curve. It is also known as the Pareto frontier, describing the boundary of achievable outcomes in a system under specific constraints.

	Figure \ref{f2} plots the Pareto curves obtained for $\zeta=$0.1, 0.4, and 1, respectively, and displaying also the Pareto points obtained with $\beta=$ 0, 0.0001, 0.001, 0.01, 0.1, and 1 (see the markers along the curves). The performance of the generalized Barker and the P3 code is also provided. The curves indicate a trade-off between the detection performance and the delay-Doppler estimation accuracy, emphasizing the role of the weight $\beta$ in the determination of the Pareto point and the cost paid for increasing one component while penalizing the other. In particular, better detection performance arising by increasing $\beta$, will inevitably bring about a decrease in estimation performance, and vice versa. In addition, the curve gradually expands outward as the similarity parameter $\zeta$ increases because of the larger feasible set. Finally, the designed code is superior to the existing non-adaptive counterparts in terms of SINR, and can also achieve a better estimation performance suitably selecting $\beta$; in particular for the case study at hand both P3 and generalized Barker are not Pareto optimal. 
	
	\begin{figure}[htbp]
		\centering
		\centering
		\includegraphics[width=3.3in]{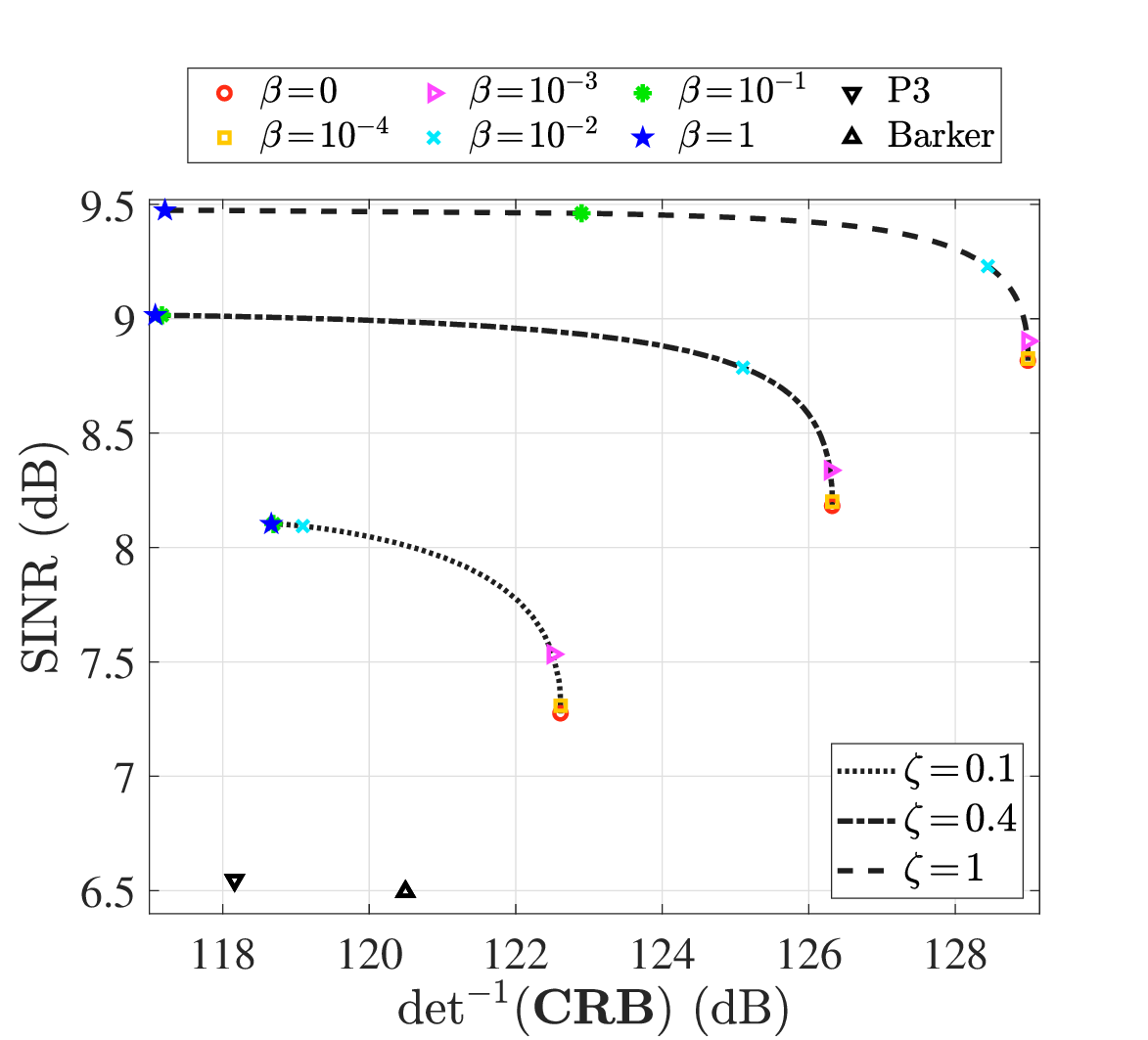}
		\caption{ \small Pareto curves for $\zeta=$ 0.1, 0.4, and 1.}\label{f2}
	\end{figure}
	
	Figures \ref{f3}\subref{f3a}, \subref{f3b}, \subref{f3c}, and \subref{f3d} depict the CRB for delay estimation, the CRB for Doppler estimation, detection probability, and the CRB determinant versus beta, respectively, for the optimized radar code (under $\zeta=$ 0.4), the reference code, and benchmarks. The curves show that, increasing $\beta$, higher values of $P_\rmd$, $\det({\textbf{CRB}_{\tau f_\rmd}(\bc)} )$, and ${\text {CRB}}_f$ are obtained, but ${\text {CRB}}_\tau$ values becomes lower and lower. This was expected, as a higher $\beta$ imposes a stronger relevance of the SINR term in the objective function leading to a possible improvement of the detection probability at the expense of the estimation accuracy, i.e.,  $P_\rmd$ and $\det({\textbf{CRB}_{\tau f_\rmd}(\bc)} )$ are increased accordingly. On the other hand, $\text{CRB}_{\tau}$ defined in \eqref{CRBt} is proportional to the inverse of the SINR, thus it decreases as $\beta$ increases. In particular, for $\beta=1$, the ${\text {CRB}}_\tau$ achieved via our devised code is close but not smaller than the benchmark, because only the energy constraint is forced into the optimization problem to obtain the benchmark; the optimized CRB determinant is worse than that of the reference code since the CRB determinant disappears in the objective function for $\beta=$ 1. In addition, the code designed with $\beta=$ 0 yields the best Doppler estimation among the different values of $\beta$, albeit slightly inferior to the benchmark code owing to the similarity restriction. Additionally, the $P_\rmd$ of the designed sequence experiences a degradation, but it still remains greater than that of the reference code (P3 code) for the same value of $|\alpha|^2$. This is explainable because, even though the SINR is removed from the objective function when $\beta=0$, there are still positively correlated terms with SINR in $\det^{-1}({\textbf{CRB}_{\tau f_\rmd}(\bc)} )$ as described in \eqref{detCRB}, hence maximizing $\det^{-1}({\textbf{CRB}_{\tau f_\rmd}(\bc)} )$ as a by-product improves also the SINR. Finally, it is evident that the higher the echo power, i.e., $|\alpha|^2$, the better both detection probability and parameter estimation accuracy.

	The effects of the similarity parameter $\zeta$ are analyzed in Figs. \ref{f4}\subref{f4a}, \subref{f4b}, \subref{f4c}, and \subref{f4d} for $\beta=$ 0.01. The curves show that increasing $\zeta$ leads to lower $\det({\textbf{CRB}_{\tau f_\rmd}(\bc)} )$, in general lower individual CRBs, and higher $P_\rmd$ because of the larger size of the similarity region and the presence of positively correlated terms with SINR in $\det^{-1}({\textbf{CRB}_{\tau f_\rmd}(\bc)} )$. Indeed, for $\zeta=$ {0.01}, 0.1, the attained Doppler accuracies are worse than that of the reference code, whereas the delay accuracy and  $\det({\textbf{CRB}_{\tau f_\rmd}(\bc)} )$ are better. This happens because the CRB determinant is optimized instead of the individual CRBs for delay or Doppler estimation. As expected, for $\zeta=$ 2, the optimized CRB determinant is slightly worse than the benchmark, because $\beta$ is small but not zero. 
	
	{To illustrate the effectiveness of the similarity constraint in preserving specific properties of the reference code, the PAPR of the designed sequence and ISL of its auto-correlation function \cite{he2012waveform} are reported in Table \ref{table2} for $\beta=$ 0.01 as a function of the similarity level. For the reference P3 code, the PAPR and ISL are 1 and -9.62 dB, respectively. As anticipated, tightening the similarity constraint (indicated by a decrease in $\zeta$) results in the PAPR and ISL of the devised sequence progressively approaching the corresponding values of the P3 code.

	\begin{table}[htbp]                                                                                                                                                                        
		\centering  
		\renewcommand{\arraystretch}{1.2}
		\caption{\small PAPR and ISL of the devised code for $\beta=$ 0.01.}\label{table2}  
		\setlength{\tabcolsep}{2mm}
		\begin{tabular}{|c|| c|c|c|c|c|c|}
			\hline  
			{$\zeta$} & {0.01} & {0.1} & {0.2} & {0.4}&1&2\\
			\cline{2-7}  
			\hline    
			\hline 
			PAPR & {1.29} & {2.09} & {2.46} & {2.71}&3.38 & 4.64\\ 
			\hline
			ISL (dB) & -8.70 & -3.28& -0.10&3.34&7.74&7.83\\ 
			\hline
		\end{tabular}
	\end{table}}
			
	\begin{figure*}[htbp]
		\centering
		\subfigure[]{\label{f3a}
			\includegraphics[width=3in]{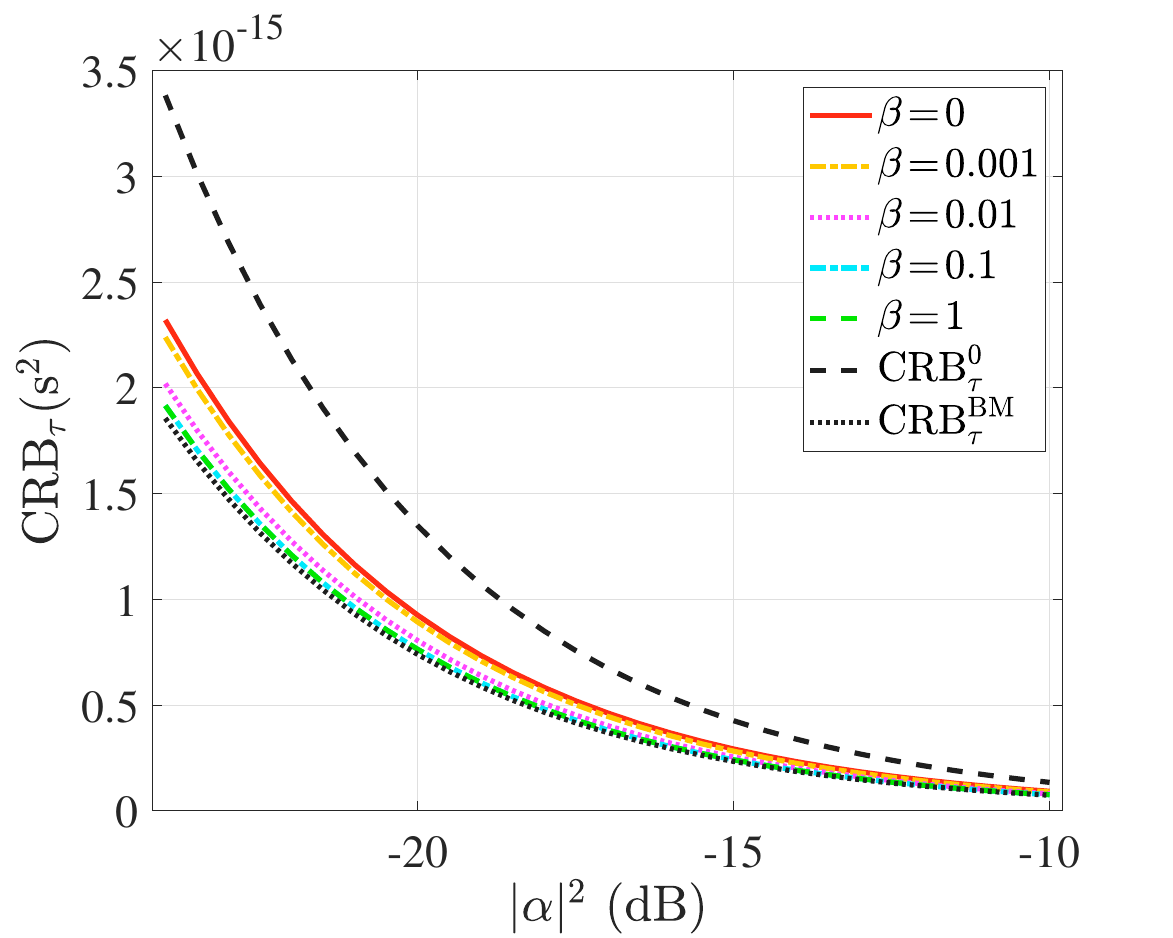}
		}\hspace{-6mm}
		\subfigure[]{\label{f3b}
			\includegraphics[width=3in]{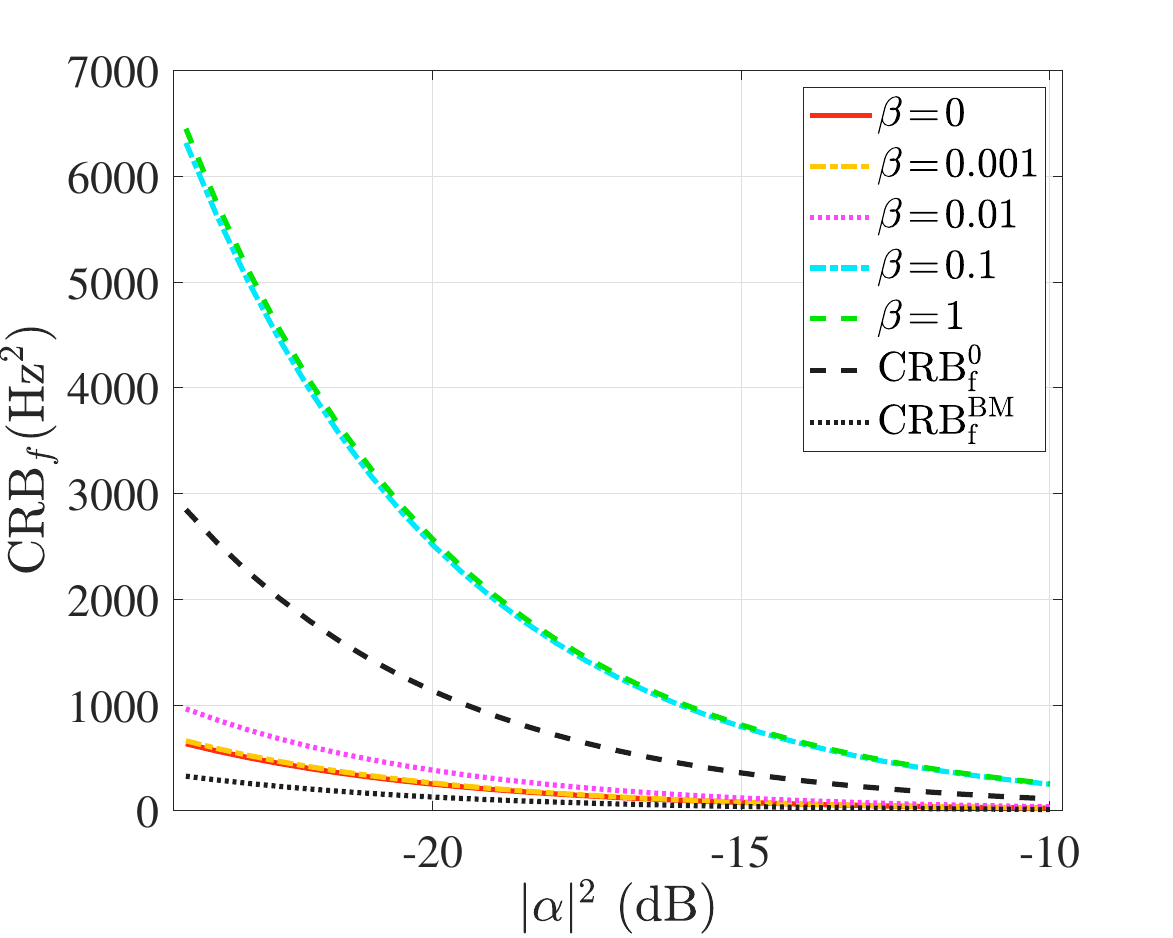}
		}\hspace{-6mm}\vspace*{-3.6mm} \\
		\subfigure[]{\label{f3c}
			\includegraphics[width=3in]{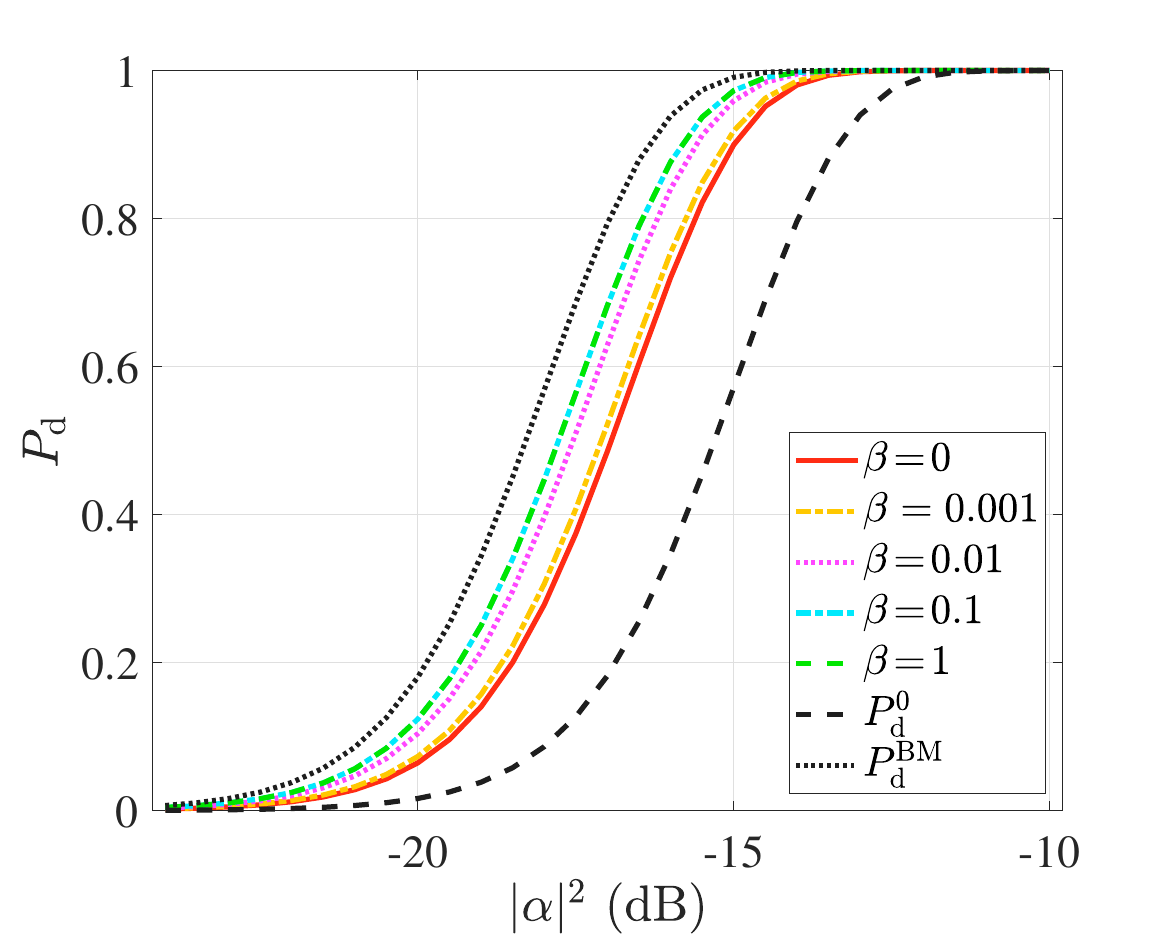}
		}\hspace{-6mm}
		\subfigure[]{\label{f3d}
			\includegraphics[width=3in]{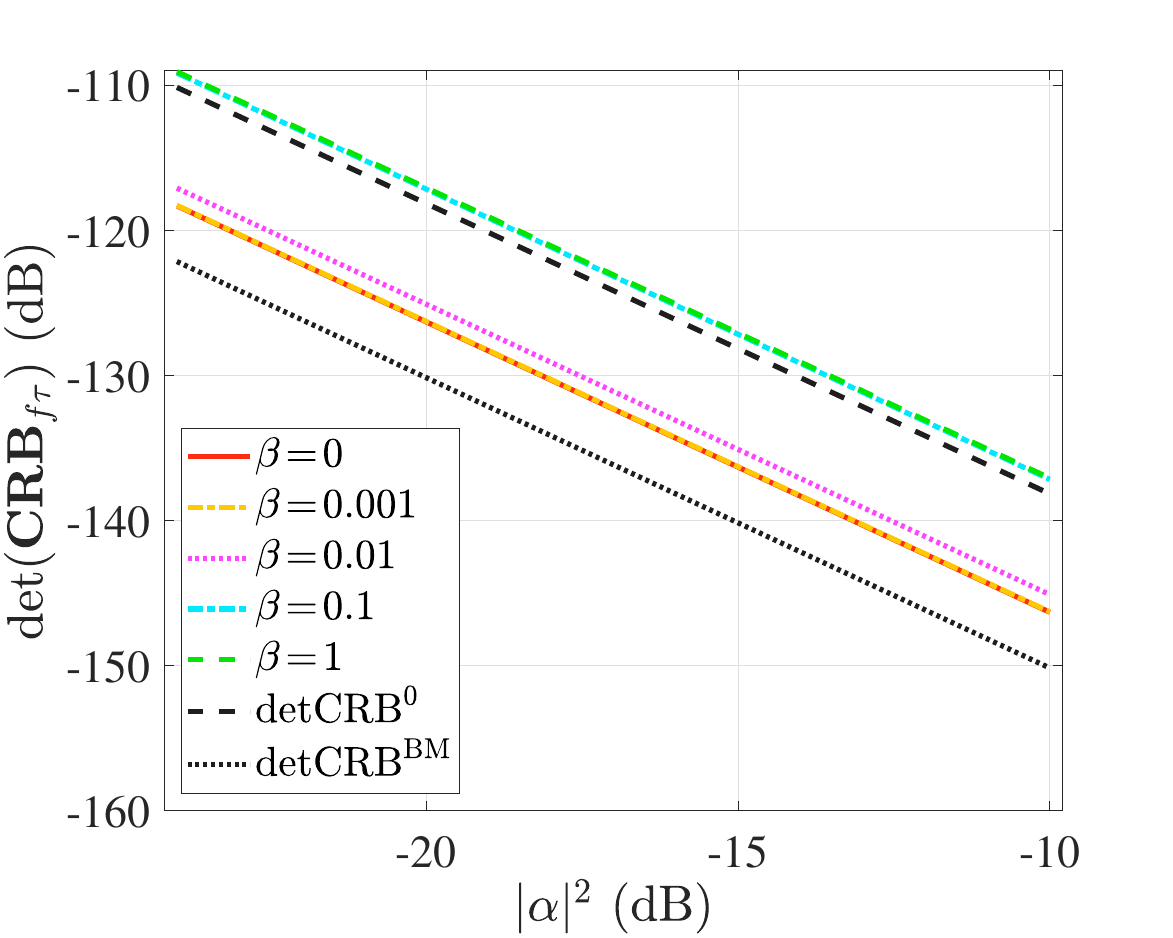}
		}\hspace{-6mm}
		\vspace{-3mm}
		\caption{\small Behavior of the code designed with $\zeta=$ 0.4 versus $|\alpha|^2$ in dB: (a) CRB for delay estimation, (b) CRB for Doppler estimation, (c) $P_\rmd$, and (d) CRB determinant.}\label{f3}
	\end{figure*}
	
	\begin{figure*}[htbp]
		\centering
		\subfigure[]{\label{f4a}
			\includegraphics[width=3in]{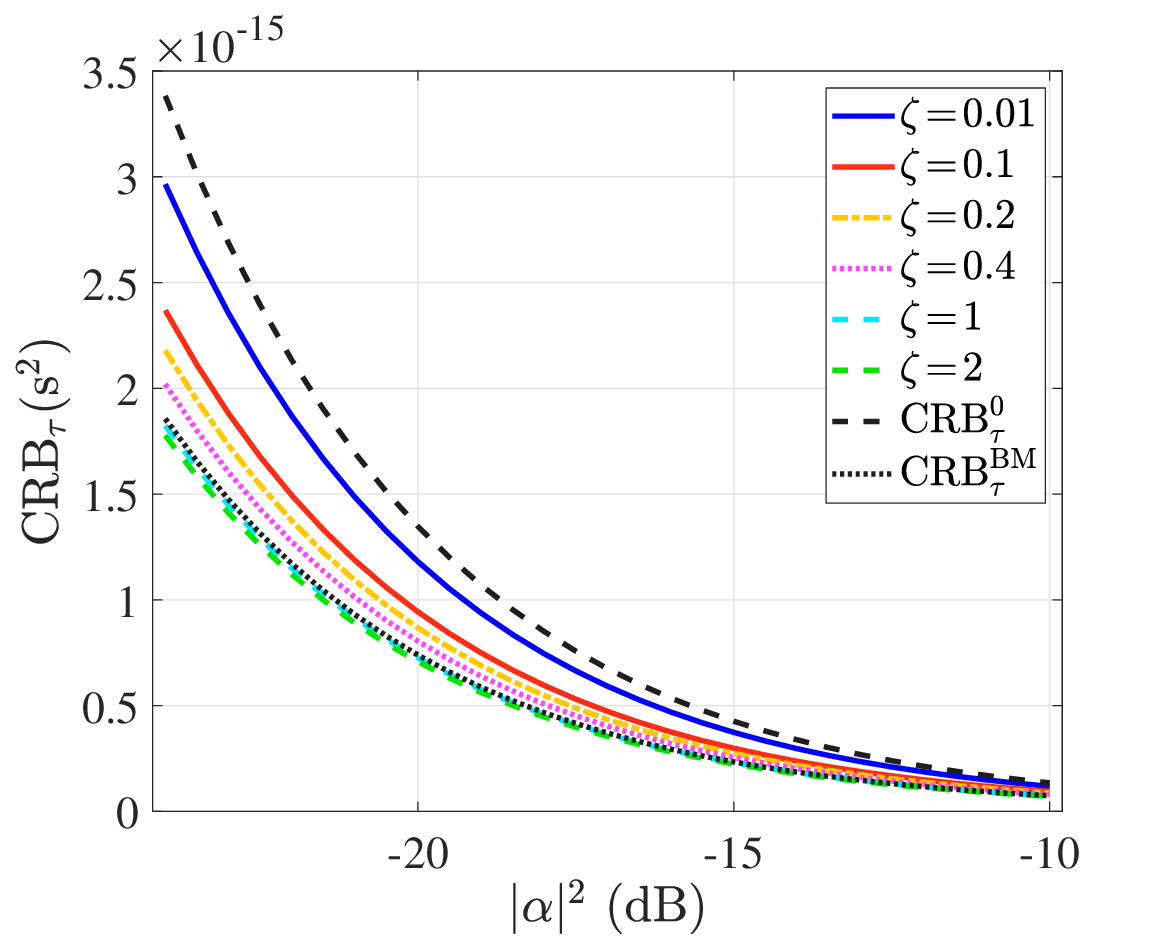}
		}\hspace{-6mm}
		\subfigure[]{\label{f4b}
			\includegraphics[width=3in]{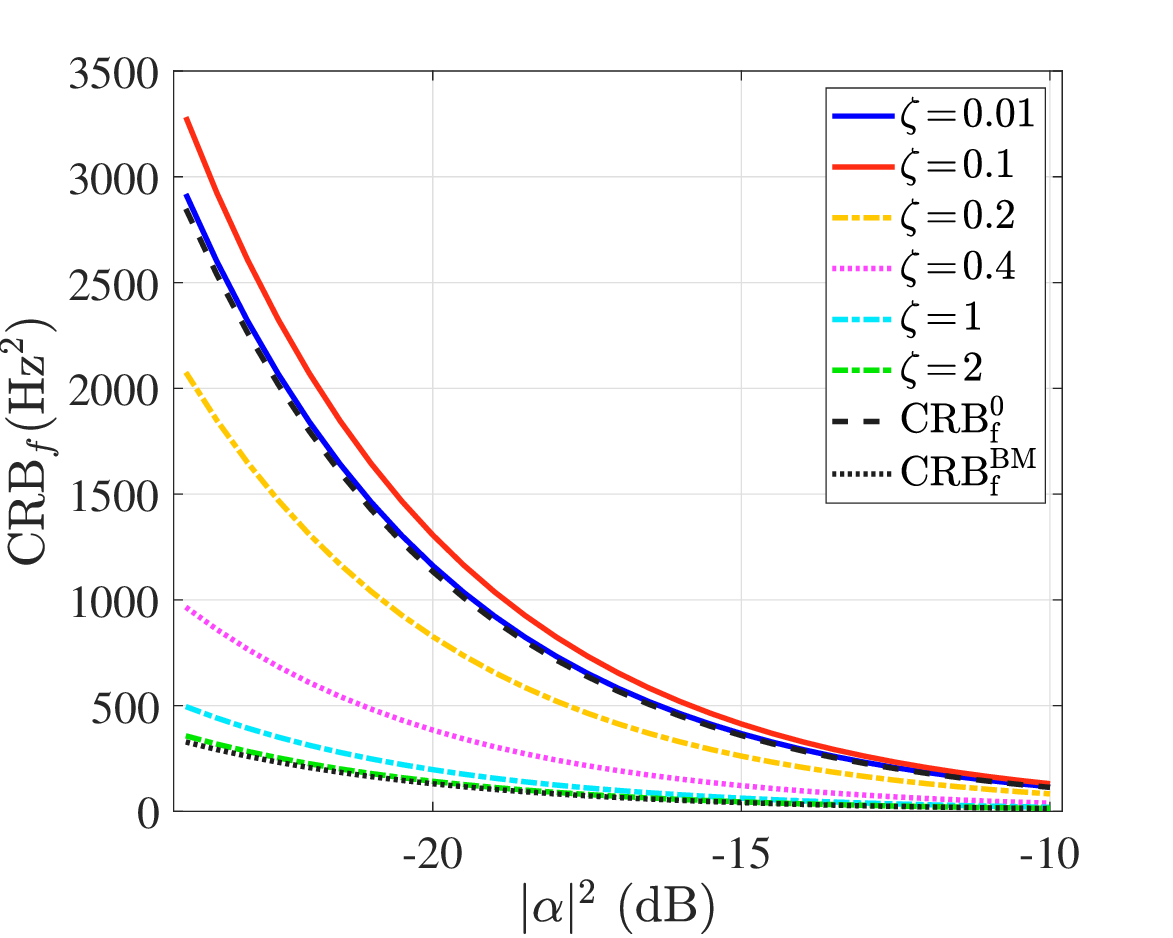}
		}\hspace{-6mm}\vspace{-3.6mm} \\
		\subfigure[]{\label{f4c}
			\includegraphics[width=3in]{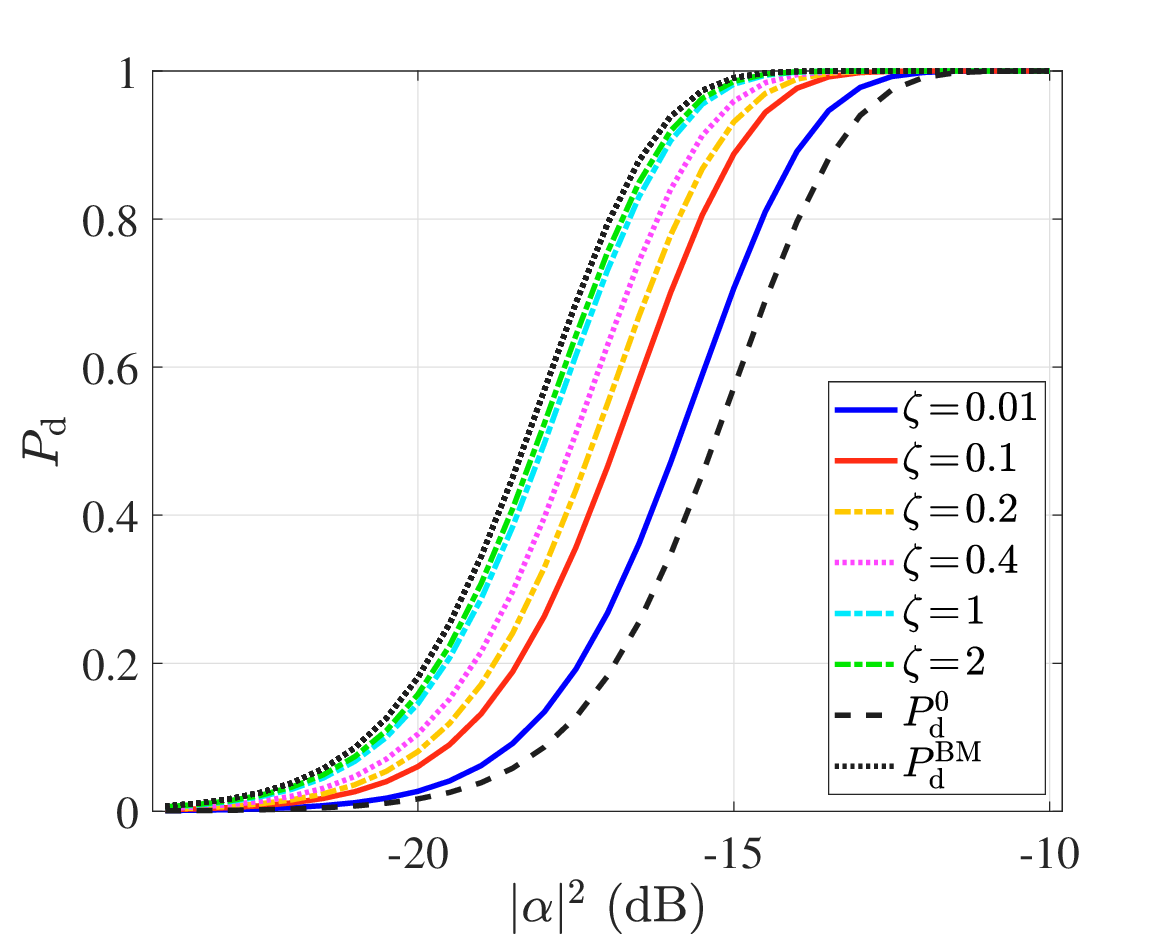}
		}\hspace{-6mm}
		\subfigure[]{\label{f4d}
			\includegraphics[width=3in]{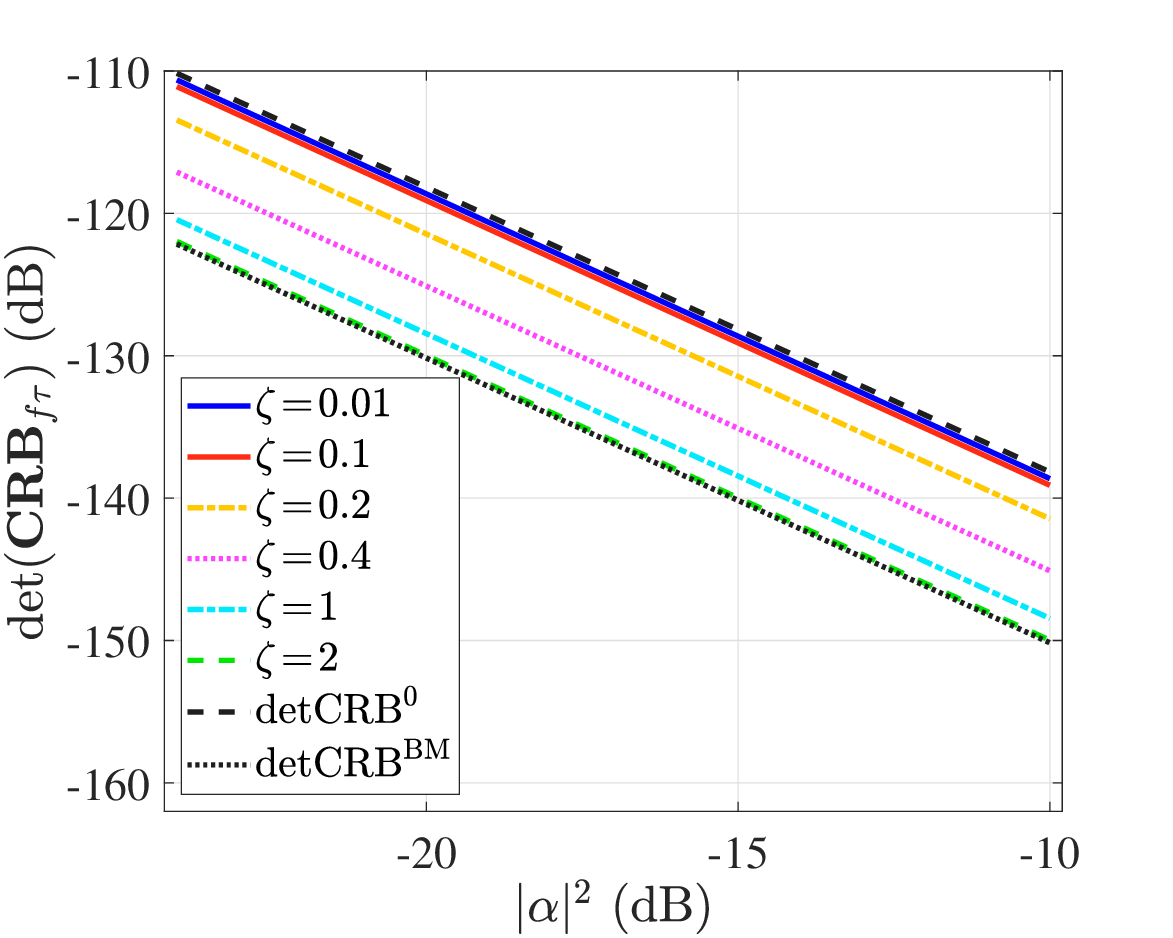}
		}\hspace{-6mm}
		\vspace{-3mm}
		\caption{\small Behavior of the code designed with $\beta=$ 0.01 versus $|\alpha|^2$ in dB: (a) CRB for delay estimation, (b) CRB for Doppler estimation, (c) $P_\rmd$, and (d) CRB determinant.}\label{f4}
	\end{figure*}

	\subsection{Doppler Sensitivity Analysis}
	{In this subsection, the behavior of the proposed algorithm is studied when a mismatch exists between the nominal and actual Doppler frequencies. In other word, the Doppler frequency assigned in \eqref{Pareto} at the design stage may differ from the actual frequency. Two reference codes are employed assuming a nominal Doppler frequency $f_\rmd=$ 600 Hz (corresponding to a normalized Doppler frequency $\nu_\rmd = f_\rmd T_\rmr=$ 0.15).} Specifically, the generalized Barker and the P3 are used. To quantify the mismatch loss, the CRBs and the detection probability, achieved by the designed code in the presence of a target with an actual Doppler frequency $\nu\in[-0.5,0.5]$, are adopted as figure of merits. Specifically, their normalized values with respect to the corresponding nominal performance levels are considered, namely,
	\beq
		\begin{aligned}
			\overline {\text{CRB}}_\tau(\nu) &= \frac{\bc^\ctrans \bM_0(\nu) \bc}{\bc^{\star \ctrans} \bM_0(\nu_\rmd ) \bc^\star},\\[0.5em]
			\overline {\text{CRB}}_{f}(\nu) &= \frac{\bc^\ctrans \bM_0(\nu) \bc\bc^\ctrans \bM_2(\nu) \bc \!-\!|\bc^\ctrans \bM_1(\nu) \bc|^2}{\bc^{\star \ctrans} \bM_0(\nu_\rmd) \bc^\star\bc^{\star \ctrans} \bM_2(\nu_\rmd) \bc^\star \!-\!|\bc^{\star \ctrans} \bM_1(\nu_\rmd) \bc^\star|^2}, \nonumber \\[0.5em]
			\overline {P}_\rmd(\nu) &= \frac{Q\Big(\sqrt{2 |\alpha|^2N\bc^\ctrans\bM_0(\nu)\bc}, \sqrt{-2\ln P_{\text{fa}}}\Big)}{Q\Big(\sqrt{2 |\alpha|^2N\bc^{\star \ctrans}\bM_0(\nu_\rmd)\bc^{\star}}, \sqrt{-2\ln P_{\text{fa}}}\Big)}, \nonumber
		\end{aligned}
	\eeq
	where $\bc^\star$ is the devised code, whereas $\bc$ can be either the devised or the reference code, $\bM_0(\nu) \!= \!\bSigma_{\sss \rmS}^{-1} \!\odot\! \ba(\nu T_\rmr)\ba(\nu T_\rmr)^\ctrans, \bM_1(\nu) \!=\! \Diag(\bb^*) \bM_0(\nu)$, and $ \bM_2(\nu) \!=\! \bb\bb^\ctrans \odot\bM_0(\nu)$. Based on the aforementioned definitions, the threshold for a 10\% performance loss (with respect to the matched condition) is 1.1 for estimation accuracy, and 0.9 for detection performance. 
	
	Figures \ref{f5}\subref{f5a}, \subref{f5b}, and \subref{f5c} depict the normalized CRB for delay, Doppler estimation, and normalized detection probability, respectively, assuming the P3 code as similarity sequence and setting $\zeta=$ 0.4, $\beta=$ 0.01. The dashed line refers to 10\% performance loss. The curves highlight that the designed code outperforms the P3 code at $\nu=$ 0.15 in terms of CRBs and $P_\rmd$. In particular, the estimation accuracies and detection probability loss do not exceed 10\% if the  normalized target Doppler frequency falls within $[0.05, 0.25]$. {Notably, compared to the reference sequence, the bespoke code enhances estimation accuracies and detection probability at the nominal Doppler frequency, but this improvement comes at the cost of reduced detection performance and estimation accuracies in Doppler regions less critical for the target under tracking (i.e., corresponding to Doppler values far from that of the target), as well as increased PAPR and ISL as reported in Table \ref{table2}.}
	
	Letting the generalized Barker code as a reference, Figs. \ref{f6}\subref{f6a}, \subref{f6b}, and \subref{f6c} plot the estimation accuracies and detection probability curves for the devised code and the generalized Barker sequence. The nominal Doppler frequency is indicated by the vertical dash-dotted line, and the 10\% performance losses are shown as the dashed lines. The estimation accuracy and detection performance of the synthesized code is significantly better than the generalized Barker in the neighborhood of the nominal Doppler frequency. The results also emphasize that, if the absolute value of the Doppler offset (with respect to the nominal value) is smaller than or equal to 0.1, i.e., $\nu \in[0.09, 0.28]$, the designed sequence exhibits a performance loss smaller than 10\% for both estimation accuracy and detection probability. 
	
	In conclusion, utilizing both the P3 and the generalized Barker code as reference signatures, the designed sequences provide an improved detection probability and estimation accuracy at the nominal Doppler frequency, while maintaining a satisfactory performance for the normalized Dopplers whose deviations are in the order of 0.1 from the nominal value.
	
	\begin{figure*}[htbp]
		\centering
		\subfigure[]{\label{f5a}
			\includegraphics[width=2.2in]{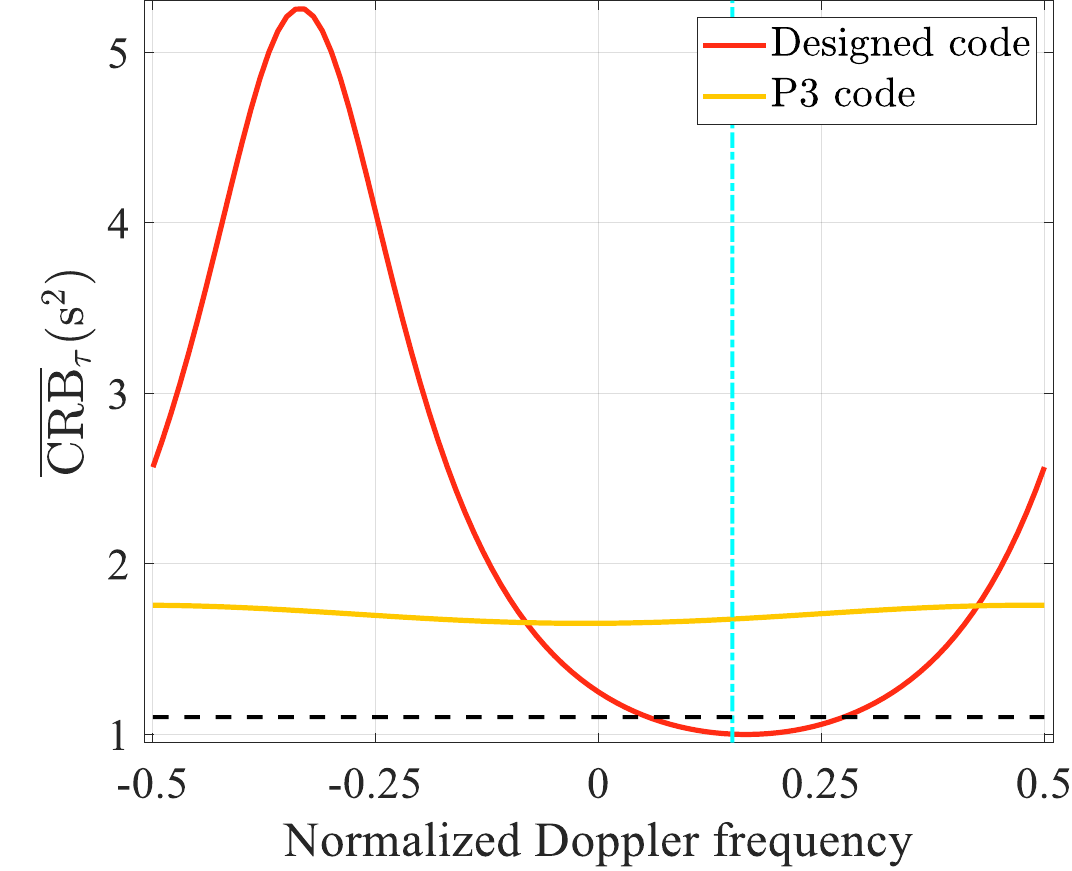}
		}\hspace{-3mm}
		\subfigure[]{\label{f5b}
			\includegraphics[width=2.2in]{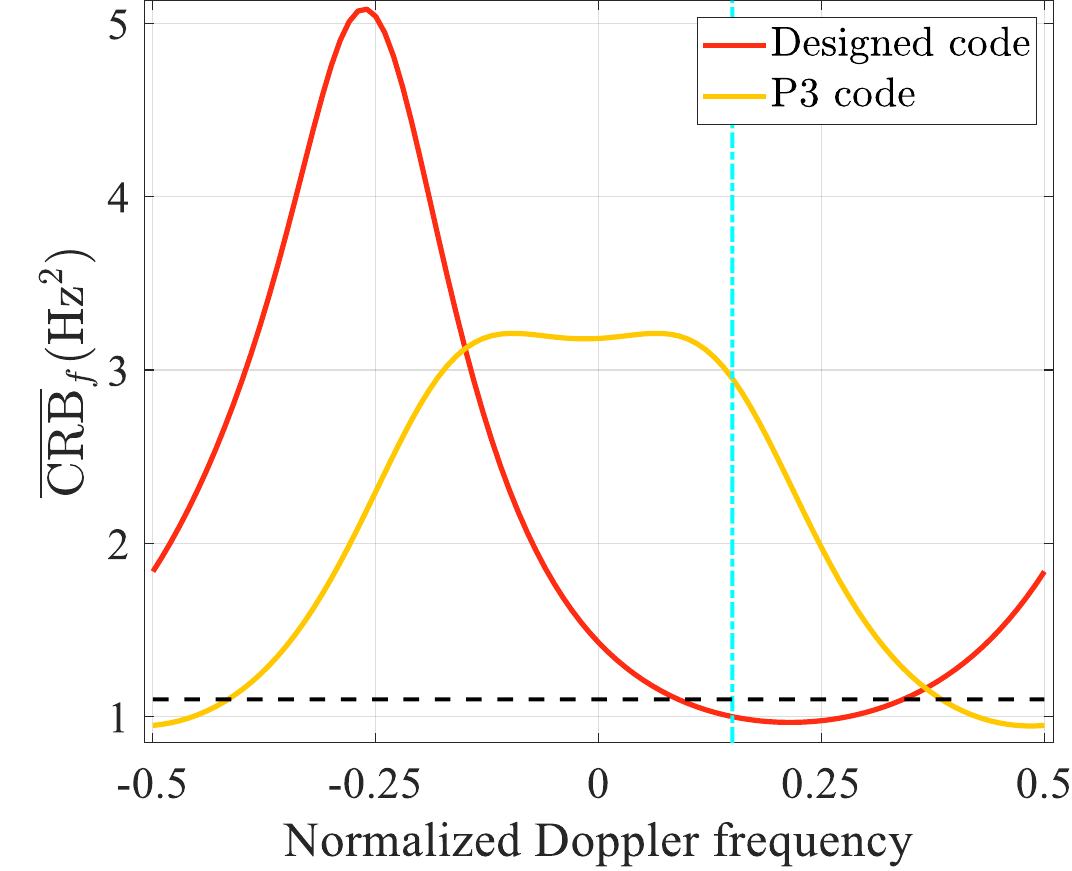}
		}\hspace{-3mm}
		\subfigure[]{\label{f5c}
			\includegraphics[width=2.2in]{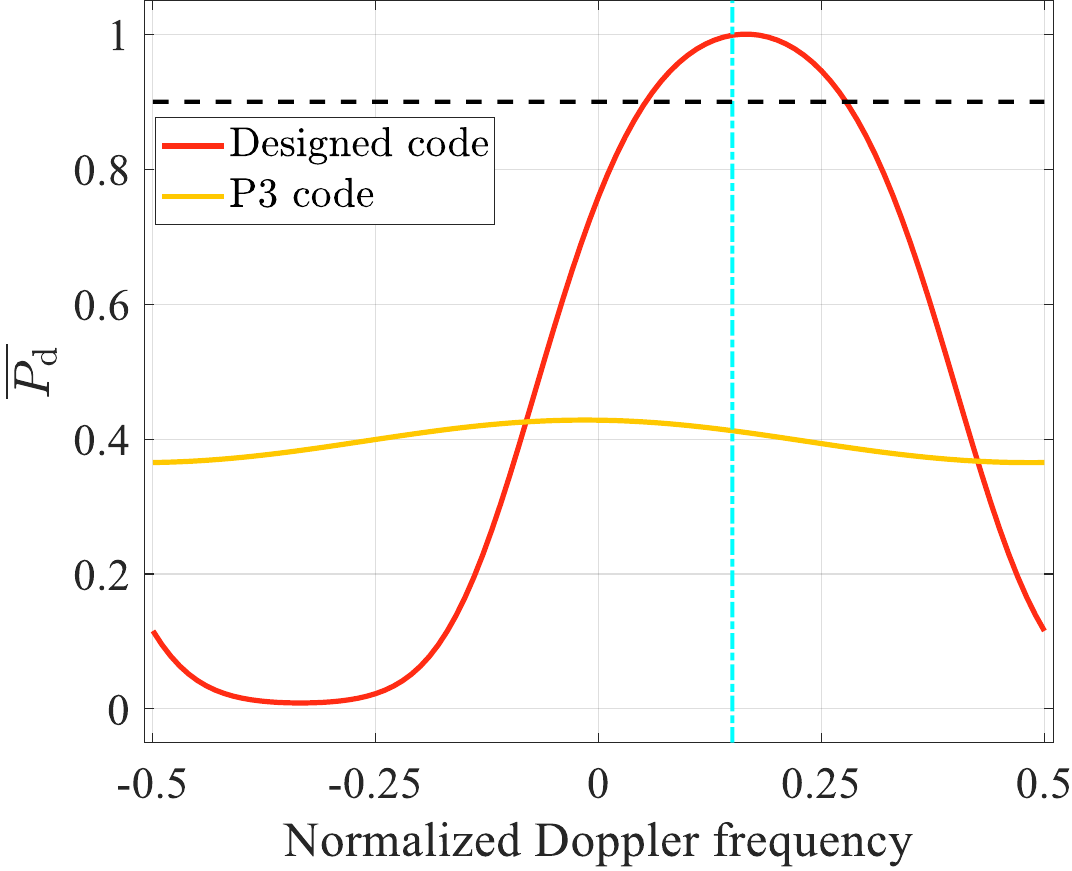}
		}\hspace{-3mm}
		\caption{ \small Behavior of the synthesized sequence obtained with $\beta\!=$ 0.01, $\zeta=$ 0.4, $\nu_\rmd=$ 0.15, and P3 as reference, versus the normalized Doppler frequency: (a) normalized CRB for delay estimation, (b) normalized CRB for Doppler estimation, and (c) normalized $P_\rmd$.}\label{f5}
	\end{figure*}
	
	\begin{figure*}[htbp]
		\centering
		\subfigure[]{\label{f6a}
			\includegraphics[width=2.2in]{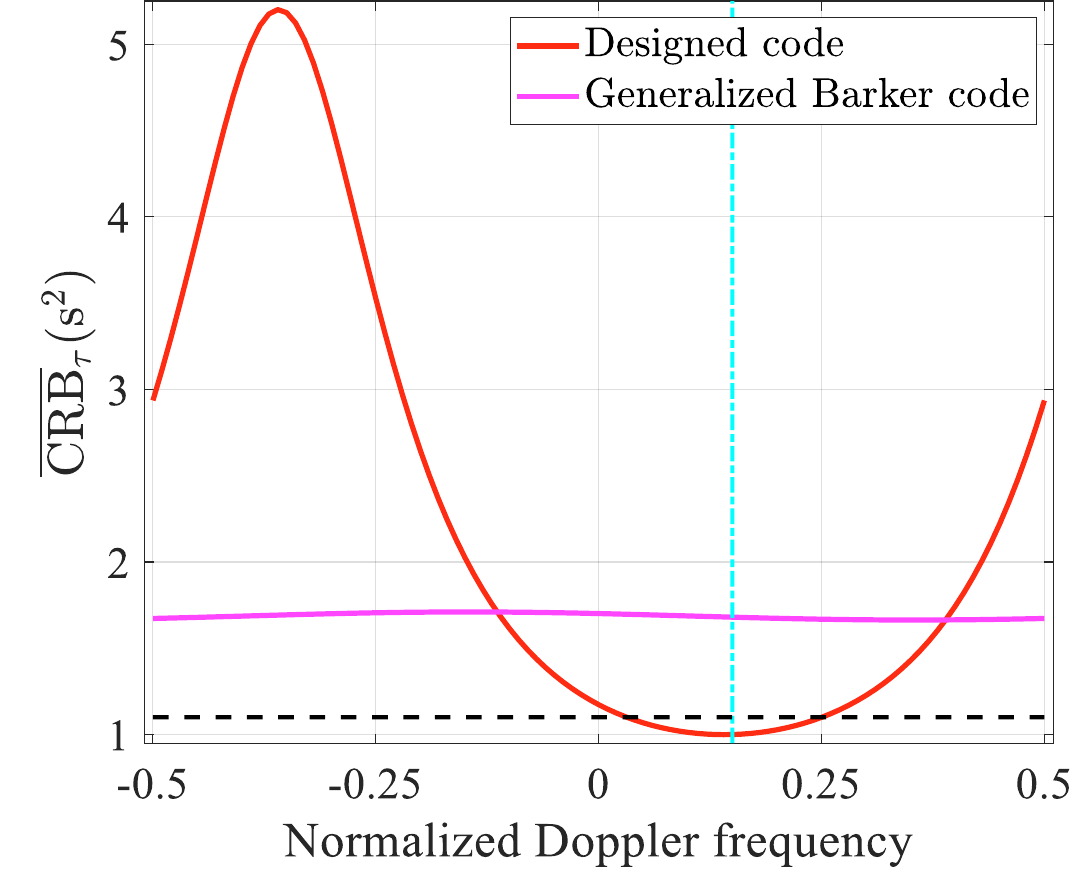}
		}\hspace{-3mm}
		\subfigure[]{\label{f6b}
			\includegraphics[width=2.2in]{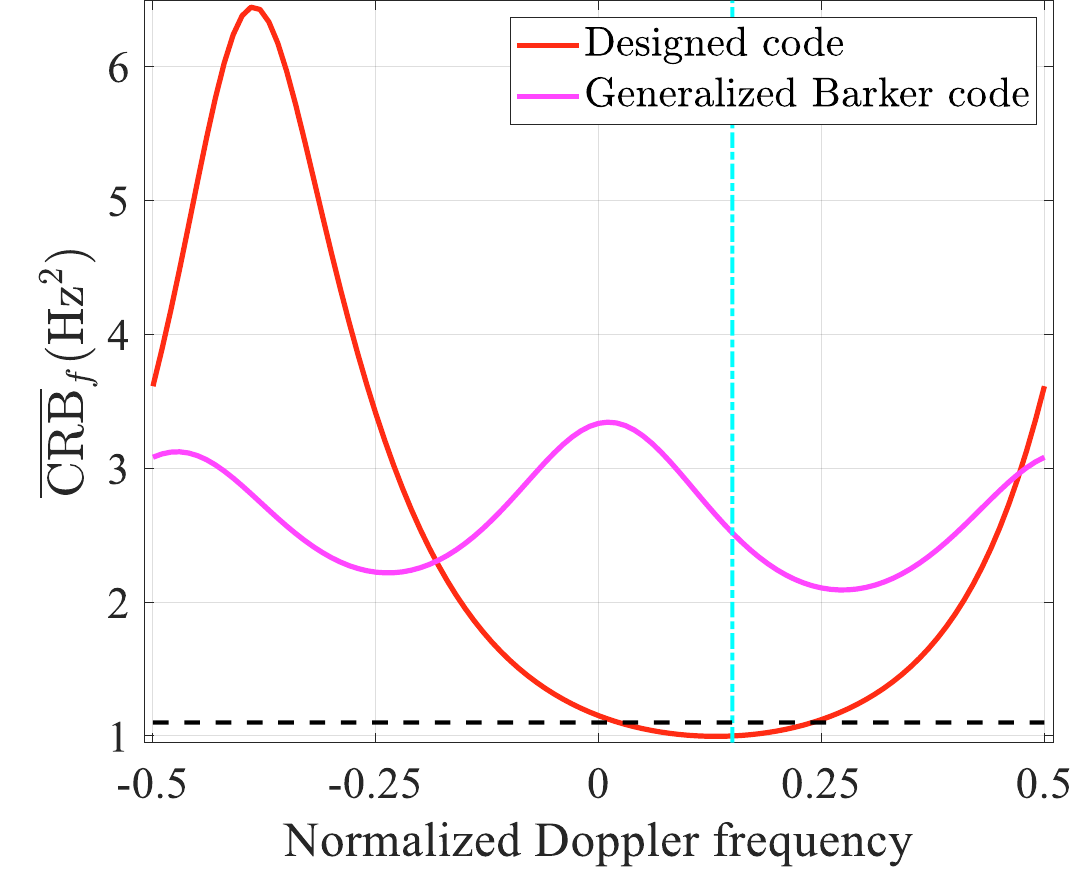}
		}\hspace{-3mm}
		\subfigure[]{\label{f6c}
			\includegraphics[width=2.2in]{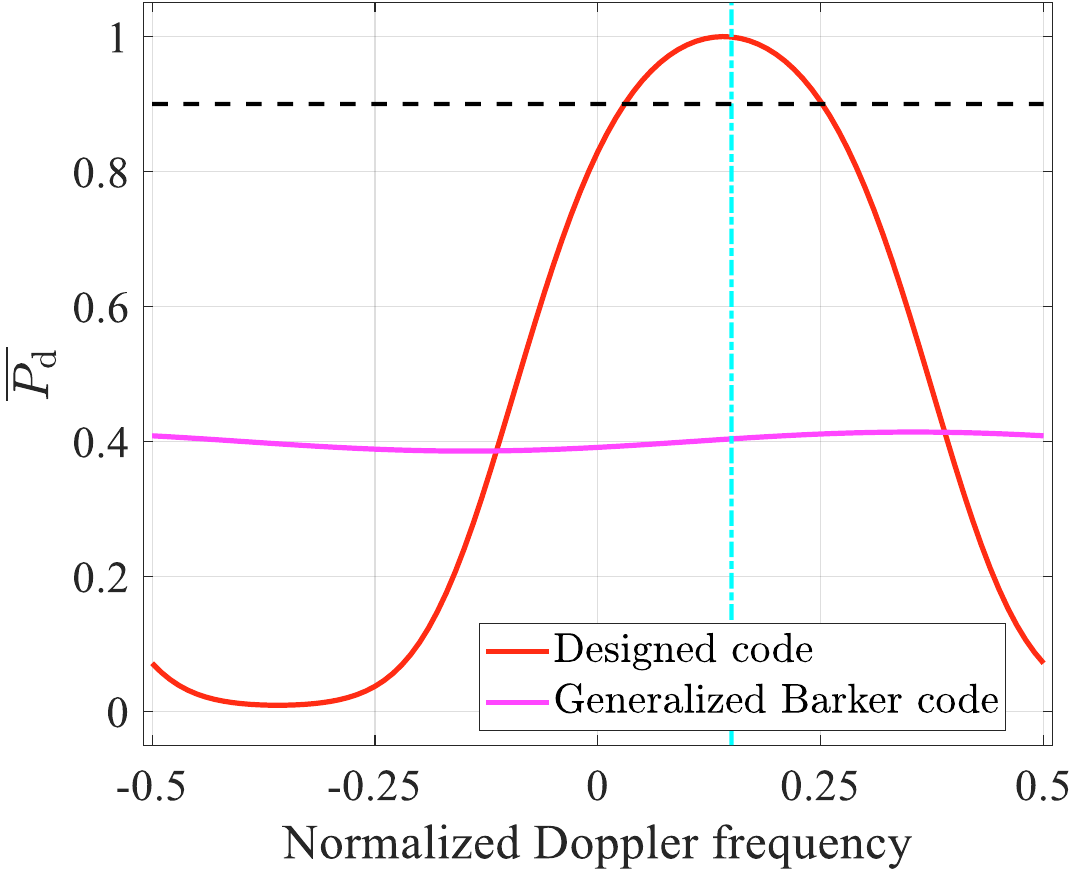}
		}\hspace{-3mm}
		\caption{ \small Behavior of the devised sequence obtained by $\beta\!=\!$ 0.01, $\zeta\!=\!$ 0.4, $\nu_\rmd\!=\!$ 0.15, and generalized Barker as reference, versus the normalized Doppler frequency: (a) normalized CRB for delay estimation, (b) normalized CRB for Doppler estimation, and (c) normalized $P_\rmd$.}\label{f6}
	\end{figure*}
	
	\section{Conclusion}\label{SecV}
	The synthesis of radar encoded waveforms under energy and similarity constraints in the presence of colored Gaussian disturbance to improve track maintenance has been explored in this paper. The design criterion relies on the joint constrained optimization of the detection performance and CRBs on delay and Doppler estimation accuracies. Specifically, the problem has been formulated as a non-convex, multi-objective optimization program subject to two quadratic constraints. With the guideline of the scalarization procedure, tensor-based relaxation, and the MBI framework, an iterative procedure with polynomial complexity and ensured convergence properties has been developed to design the slow-time code. In particular, the fourth order conjugate super-symmetric tensor for the real-valued complex inhomogeneous quartic polynomial involved in the tensor-based relaxation has been constructed. Besides, each inner problem involved in the MBI iteration has been shown to enjoy hidden-convexity and remarkably a closed-form solution. Notably, the computational complexity of the procedure is linear with the iteration number and quadratic with the radar code length. Performance evaluation at the analysis stage has included assessing the monotonicity of the objective value, detection performance, CRBs for delay and Doppler estimation accuracies, and sensitivity to target Doppler. Moreover, the Pareto curves have been analyzed to illustrate the impact of Pareto weights on the trade-off between detection performance and delay-Doppler estimation accuracy. In summary, the main contributions of the paper are:
	\begin{enumerate}[a)]
			\item the introduction of a slow-time code synthesis technique for jointly optimizing SINR, delay-Doppler accuracy under colored Gaussian disturbance;
			\item the development of a polynomial-time iterative procedure combining scalarization, tensor-based relaxation, and MBI approaches for the  multi-objective optimization problem, as well as the investigation of convergence properties of the relaxed procedure;
			\item the construction of the fourth order tensor for the involved real-valued complex inhomogeneous quartic polynomial;
			\item the presentation of numerical results to validate the theoretical achievements.
	\end{enumerate}
	
	
	Potential future research avenues might include the extension of the approach to {multi-target scenario\cite{fan2024airborne}, signal-dependent clutter scenario\cite{sen2010ofdm,aubry2013knowledge}}, space-time processing in multi-channel systems (e.g. frequency diversity arrays \cite{lan2023adaptive, lan2024fda}, MIMO\cite{cui2014mimo, yu2022integrated}, and polarimetric radars\cite{cheng2017robust, jin2021quasi}). Moreover, it would be worth to account for some additional radar measurement impairments (such as phase noise\cite{aubry2016radar1, aubry2016radar2}) as well as deal with spectrum contested and congested environments via appropriate and specific constraints\cite{aubry2020design}.
	
	\section*{Acknowledgment}
		This research activity has been conducted during the visit of Tao Fan at the Universit\`a degli Studi di Napoli “Federico II”, DIETI under the local supervision of Prof. A. De Maio and  Prof. A. Aubry. 
	
			\appendix
	\subsection{Proof of Proposition \ref{Prop1}} \label{proofCRB}
	Since $\bmu(\bgamma)=\alpha (\ba(f_{\rm d})  \odot  \bc) \otimes  {\bs(\tau)}$, the terms $\disps\frac{\partial \bmu({\bgamma})}{\partial \widetilde \balpha^\trans}$, $\disps\frac{\partial \bmu({\bgamma})}{\partial \tau}$, and $\disps\frac{\partial \bmu({\bgamma})}{\partial f_\rmd}$ in \eqref{Dgamma} are, respectively, given by $\frac{\partial \bmu({\bgamma})}{\partial \widetilde \balpha^\trans} = [1, \jmath]\! \otimes \! [(\ba(f_{\rm d}) \! \odot \! \bc) \! \otimes \! {\bs(\tau)}]$, $\frac{\partial \bmu({\bgamma})}{\partial \tau}= \alpha (\ba(f_{\rm d}) \! \odot \! \bc) \! \otimes \!\frac{\partial  {\bs(\tau)}}{\partial \tau}$, and $\disps\frac{\partial \bmu({\bgamma})}{\partial f_\rmd} = \alpha (\ba(f_{\rm d}) \! \odot \! \bc\! \odot \! \bb) \! \otimes \! {\bs(\tau)}$, where $\bb\!=\![0,\jmath2\pi T_\rmr,\cdots,\jmath2\pi (M\!-\!1)T_\rmr]^\trans$. Then, according to the Slepian-Bangs formula\cite{stoica1997intro}, the Fisher information matrix $\mathbfcal{I}$  is
	\beq
	\mathbfcal{I}=\disps \left[ {\begin{array}{*{20}{l}}
			{\disps \mathbfcal{I}_{ \widetilde \balpha\widetilde \balpha}}&{\mathbfcal{I}_{\tau { \widetilde \balpha}}^\trans}&{\mathbfcal{I}_{{ f}{\widetilde \balpha}}^\trans}\\
			{\mathbfcal{I}_{\tau { \widetilde \balpha}}}&{\mathcal{I}_{\tau \tau}}&{\mathcal{I}_{f\tau }^\trans}\\
			{\mathbfcal{I}_{f{ \widetilde \balpha}}}&{\mathcal{I}_{ f\tau }}&{\mathcal{I}_{ff}}
	\end{array}} \right],
	\eeq
	where 
	\begin{subequations}
		\begin{eqnarray}
			\!\!\!\mathbfcal{I}_{\widetilde \balpha\widetilde\balpha} &\!\!\!\!\!\!=\!\!\!\!\!\!& 2\|{\bs(\tau)}\|^2(\ba(f_{\rm d}) \! \odot \! \bc)^\ctrans\bSigma_{\rmt}^{-1}(\ba(f_{\rm d}) \! \odot \! \bc)  \bI_2, \nonumber \\
			\mathbfcal{I}_{\tau \widetilde \balpha}\;&\!\!\!\!\!=\!\!\!\!\!\!& 2\Re\!\left\{\![1, \jmath]\! \otimes \! \alpha^*(\ba(f_{\rm d}) \! \odot \! \bc)^\ctrans\bSigma_{\rmt}^{-1}(\ba(f_{\rm d}) \! \odot \! \bc) \frac{\partial^\ctrans \! {\bs(\tau)}}{\partial \tau}{\bs(\tau)}\!\right\}\!\!, \nonumber \\
			\mathbfcal{I}_{f\widetilde\balpha}\; &\!\!\!\!\!\!=\!\!\!\!\!\!&  2\|{\bs(\tau)}\|^2\Re\big\{[1, \jmath]\! \otimes \!\alpha^*(\ba(f_{\rm d}) \! \odot \! \bc\! \odot \! \bb)^\ctrans\bSigma_{\rmt}^{-1}(\ba(f_{\rm d}) \! \odot \! \bc)\big\}, \nonumber\\
			\mathcal{I}_{\tau \tau}\;&\!\!\!\!\!\!=\!\!\!\!\!\!& 2|\alpha|^2 \Big\|\frac{\partial  {\bs(\tau)}}{\partial \tau}\Big\|^2(\ba(f_{\rm d}) \! \odot \! \bc)^\ctrans\bSigma_{\rmt}^{-1}(\ba(f_{\rm d}) \! \odot \! \bc), \nonumber\\
			\mathcal{I}_{ f\tau}\;&\!\!\!\!\!\!=\!\!\!\!\!\!& 2|\alpha|^2\Re\!\left\{\!(\ba(f_{\rm d}) \! \odot \! \bc)^\ctrans\bSigma_{\rmt}^{-1}(\ba(f_{\rm d}) \! \odot \! \bc \! \odot \! \bb)  \frac{\partial^\ctrans\!  {\bs(\tau)}}{\partial \tau} {\bs(\tau)}\!\right\}\!\!, \nonumber\\
			\mathcal{I}_{ff}\;&\!\!\!\!\!\!=\!\!\!\!\!\!& 2|\alpha|^2 \|{\bs(\tau)}\|^2  (\ba(f_{\rm d}) \! \odot \! \bc \! \odot \! \bb) ^\ctrans\bSigma_{\rmt}^{-1}(\ba(f_{\rm d}) \! \odot \! \bc \! \odot \! \bb). \nonumber
		\end{eqnarray}
	\end{subequations}
	
	As a consequence, the CRB for the Doppler shift and two-way delay is 
	\beq \label{CRBft}
	\begin{aligned}
		\!\!\!\!\textbf{CRB}_{\tau f_\rmd}(\bc)\!\!=\!\!\left\{\!\left[ {\begin{array}{*{20}{l}}
				\!\!{\mathcal{I}_{\tau \tau}}&\!\!{\mathcal{I}_{f \tau }^\trans}\!\!\\
				\!\!{\mathcal{I}_{f \tau }}&\!\!{\mathcal{I}_{ff}}\!\!
		\end{array}} \right]\!-\! \mathbfcal{J}\mathbfcal{I}_{\widetilde\balpha\widetilde\balpha}^{-1}\mathbfcal{J}^\trans\right\}^{\!\!\!-1}\!\!\!\!=\!\!\left[ {\begin{array}{*{20}{l}}
				\!\!{\mathcal{\tilde I}_{\tau \tau}}&\!\!{\mathcal{\tilde I}_{f \tau }^\trans}\!\!\\
				\!\!{\mathcal{\tilde I}_{f \tau }}&\!\!{\mathcal{\tilde I}_{ff}}\!\!
		\end{array}} \right]^{\!-\!1}
	\end{aligned}\!\!\!\!\!,
	\eeq
	where $\mathbfcal{J}=[\mathbfcal{I}_{\tau\widetilde\balpha}^\trans, \mathbfcal{I}_{f\widetilde\balpha}^\trans]^\trans$. It follows that 
	\beq\label{JIJ}
	\begin{aligned}
		\!\!\!\!\mathbfcal{J}\mathbfcal{I}_{\widetilde\balpha\widetilde\balpha}^{-1}\mathbfcal{J}^\trans=&\frac{1}{2\|{\bs(\tau)}\|^2(\ba(f_{\rm d}) \! \odot \! \bc)^\ctrans\bSigma_{\rmt}^{-1}(\ba(f_{\rm d}) \! \odot \! \bc)   }\\
		&\cdot \left[ {\begin{array}{*{20}{l}}
				{\Re\{P_{\tau\widetilde\balpha}P_{\tau\widetilde\balpha}^*\}}&{\Re\{P_{\tau\widetilde\balpha}P_{f\widetilde\balpha}^*\}}\\
				{\Re\{P_{f\widetilde\balpha}P_{\tau \widetilde\balpha}^*\}}&{\Re\{P_{f\widetilde\balpha}P_{f\widetilde\balpha}^*\}}
		\end{array}} \right]
	\end{aligned},
	\eeq
	where $P_{\tau \widetilde\balpha} = 2\alpha^*(\ba(f_{\rm d}) \! \odot \! \bc)^\ctrans\bSigma_{\rmt}^{-1}(\ba(f_{\rm d}) \! \odot \! \bc) \frac{\partial^\ctrans\!  {\bs(\tau)}}{\partial \tau}{\bs(\tau)}$ and $P_{f\widetilde\balpha} = 2\|{\bs(\tau)}\|^2\alpha^*(\ba(f_{\rm d}) \! \odot \! \bc\! \odot \! \bb)^\ctrans\bSigma_{\rmt}^{-1}(\ba(f_{\rm d}) \! \odot \! \bc)$.
	Substituting equation \eqref{JIJ} into equation \eqref{CRBft}, yields $\mathcal{\tilde I}_{f \tau }=0$, $
	\mathcal{\tilde I}_{\tau \tau} = 2|\alpha|^2(\ba(f_{\rm d}) \! \odot \! \bc)^\ctrans \bSigma_{\rmt}^{-1} (\ba(f_{\rm d}) \! \odot \! \bc)\frac{\partial^\ctrans\!  {\bs(\tau)}}{\partial \tau}\bA_{\tau \tau}(\bgamma, \bc)\frac{\partial  {\bs(\tau)}}{\partial \tau}$, and $\mathcal{\tilde I}_{ff}= 2|\alpha|^2 \|{\bs(\tau)}\|^2\nonumber (\ba(f_{\rm d}) \! \odot \! \bc \! \odot \! \bb) ^\ctrans \bA_{ff}(\bgamma, \bc) (\ba(f_{\rm d}) \! \odot \! \bc \! \odot \! \bb)$, where $\bA_{\tau \tau}(\bgamma, \bc) =   \bI_N - \frac{ {\bs(\tau)} {\bs^\ctrans(\tau)}}{\|{\bs(\tau)}\|^2}$ and $\bA_{ff}(\bgamma, \bc) \!=\!\bSigma_{\rm S}^{-1}\!-\!\frac{\bSigma_{\rm S}^{-1} (\ba(f_{\rm d})  \odot  \bc) (\ba(f_{\rm d})  \odot  \bc)^\ctrans \bSigma_{\rm S}^{-1}}{(\ba(f_{\rm d})  \odot  \bc)^\ctrans \bSigma_{\rm S}^{-1} (\ba(f_{\rm d})  \odot  \bc)}$.
	
	Assuming $\tau_0 \approx \tau$ in \eqref{stau2}, one has \cite{dogandzic2001cramer} $\big\|\frac{\partial  {\bs(\tau)}}{\partial \tau}\big\|^2\!\approx\!\frac{1}{\Delta t}\int_{0}^{T_\rmp} \frac{4\pi^2B^2}{T_\rmp^2}\big(t\!-\!\frac{T_\rmp}{2}\big)^2 dt\!=\!\frac{\pi^2B^2N}{3}$ and $\frac{\partial^\ctrans\!  {\bs(\tau)}}{\partial \tau}{\bs(\tau)}\!\approx\!-\!\frac{1}{\Delta t}\int_{0}^{T_\rmp} \big(\frac{\partial s(t)}{\partial t}\big)^*s(t) dt\!=\!0$. Then, $\mathcal{\tilde I}_{\tau \tau}$ and $\mathcal{\tilde I}_{f f}$ can be simplified as $\mathcal{\tilde I}_{\tau \tau} = \frac{2|\alpha|^2 \pi^2 B^2 N}{3} \bc^\ctrans \bM_0 \bc$ and $\mathcal{\tilde I}_{ff} = 2|\alpha|^2 N \left[\bc^\ctrans \bM_2 \bc -  \frac{|\bc^\ctrans \bM_1 \bc|^2}{\bc^\ctrans \bM_0 \bc }\right]$, respectively, where $\bM_0 = \bSigma_{\rmt}^{-1} \odot \ba(f_{\rmd})\ba(f_{\rm d})^\ctrans$, $\bM_1 = \Diag(\bb^*)\bM_0$, and $ \bM_2 =\bb\bb^\ctrans \odot\bM_0 $. 
	
	Therefore, the CRB matrix for $\tau$ and $f_\rmd$  is 
	\beq
	\textbf{CRB}_{\tau f_\rmd}(\bc) =\frac{1}{2|\alpha|^2N}\left[\!\!\! {\begin{array}{*{20}{c}}
			\!{ \frac{\pi^2 B^2}{3}\bc^\ctrans \bM_0 \bc}&\!\!\!{0}\\
			{0}\!\!\!&\!\!\!{\bc^\ctrans \bM_2 \bc \!-\! \displaystyle \frac{|\bc^\ctrans \bM_1 \bc|^2}{\bc^\ctrans \bM_0 \bc }  }
	\end{array}} \!\!\!\right]^{-1}\!\!\!\!\!\!\!\!.
	\eeq
	
	\subsection{The Multi-linear Form for Describing $\widetilde f(\bc)$}\label{proof_tensor}
	
	Before proceeding further, let us introduce the definition of the fourth order conjugate super-symmetric tensor and a key corollary ( see \cite{aubry2013ambiguity,chen2022generalized} for details) that will be used in the subsequent steps.
	
	\begin{definition*}\label{tensor_def} 
		A fourth order even dimensional complex tensor $\calF\!\in\!\bbC^{2M\!\times2M\!\times\!2M\!\times\!2M}$ is conjugate super-symmetric, if (i) $\calF$ is super-symmetric, i.e., $\calF_\pi \!=\!\! \calF_{iqkl}$, $ \pi\! \in\! \Pi_4(i,q,k,l)$,	$(i,q,k,l) \!\in \!\{1,2,\cdots,2M\}^4$; and (ii) $\calF_{i_1i_2i_3i_4} =  \calF^*_{j_1j_2j_3j_4}$ if $|i_k-j_k|\!=\!M$, for all $ 1\!\leq k\! \leq 4$.
	\end{definition*}
	
	\begin{corollary}\label{coro1}
		Given a real-valued conjugate homogeneous quartic function $g(\bx) \!\!=\!\! \sum\limits_{r=1}^{R}\alpha_r\bx^\ctrans\bR_r\bx\bx^\ctrans\bR_r^\ctrans\bx$ with $\bx\!\in\! \bbC^M$, $\bR_r\!\in\!\bbC^{M\times M}$ and $\alpha_r\!\in\!\bbR$, then 
		\begin{enumerate}[(a)]
			\item\label{coro1_i} the fourth order conjugate super-symmetric tensor $\calG\in\bbC^{2M\times2M\times2M\times2M}$ corresponding to $g(\bx)$ is given by 
			\beq\label{Gpi}
			\!\!\!\!\!\!\!\!\calG_\pi \!\!=\!\! \left\{\!\!\begin{array}{lll}
				\!\!\!\displaystyle\frac{a_{iqkl}}{\text{Card}(\Pi_4\!(M\!\!+\!\!i, \!q, \!M\!\!+\!\!k, \!l))}, & {\begin{array}{l}
						\!\!\!\!\!\! {\text {if}}\; \pi \!\in\! \Pi_4\!(M\!\!+\!\!i, \!q, \!M\!\!+\!\!k, \!l), \\
						\!\!\!\!\!\!(i,\!q,\!k,\!l) \!\in\! \{1,\!2,\cdots,\!M\}^4,
				\end{array}} \\
				0, & \!\!\!\text{otherwise},
			\end{array}
			\right. \nonumber
			\eeq
			where $a_{iqkl}=\sum\limits_{r=1}^{R}\sum\limits_{(i',k')\in \Pi_2(i,k)\atop(q',l')\in \Pi_2(q,l)} \!\!\!\!\!\!\!\!\alpha_r\bR_r(i',q')\bR_r^*(l',k')$, and ensures that $g(\bx)=\calG\big( \binom{ \bx}{ \bx^*}, \binom{ \bx}{ \bx^*}, \binom{ \bx}{ \bx^*}, \binom{ \bx}{ \bx^*}\big)$.
			\item\label{coro1_ii} the multi-linear polynomial function associated to $g(\bx)$, derived from the super-symmetry of $\calG$, can be expressed as $g_{_\text{ML}}(\bx^1, \bx^2, \bx^3, \bx^4) = \frac{1}{24}\! \sum\limits_{r=1}^{R} \sum\limits_{(p_1, p_2, p_3, p_4) \atop \in \Pi_4 (1, 2, 3, 4)}\!\!\!\!\!\!\!\!\alpha_r(\bx^{p_1})^\ctrans\bR_r\bx^{p_2}$ $ (\bx^{p_3})^\ctrans\bR_r^\ctrans\bx^{p_4}$, where $g_{_\text{ML}}(\bx^1, \bx^2, \bx^3, \bx^4)\!\!=\!\!g(\bx)$, if $\bx^1\!\!=\!\!\bx^2\!\!=\!\!\bx^3\!\!=\!\!\bx^4\!\!=\!\!\bx$.
		\end{enumerate}
	\end{corollary}

	Now, observe that $\widetilde f(\bc)$ is the superposition of multiple real-valued complex quartic and quadratic functions. To proceed further, let us focus on the class of real-valued inhomogeneous complex quartic functions
	\beq\label{fI}
	f_{\rmI}(\bc) =\alpha_{0}\bc^\ctrans \bD\bc+\sum\nolimits_{j=1}^{J}\alpha_j\bc^\ctrans\bA_j\bc\bc^\ctrans\bB_j^\ctrans\bc,
	\eeq
	where $\bD \!\in\! \bbH^{M} $, $\alpha_0, \alpha_{j}\!\!\in\!\bbR$, and either $\bA_j, \bB_j \!\in\! \bbH^{M}$ or $\bA_j\!\!=\!\! \bB_j\!\!\in\!\!\bbC^{M\times M},j\!\!=\!\!1,2,\cdots,J$. Specifically, $\widetilde f(\bc)$ boils down to \eqref{fI} considering $J=3$, and taking $\alpha_0\!=\!1, \alpha_1\!=\!\mu_1, \alpha_2\!=\!1\!-\!\beta, \alpha_3\!=\!-1\!+\!\beta$, $\bD\!=\!\mu_2\bI_M\!+\!\beta\bM_0$, $\bA_1\!=\!\bB_1\!=\!\bI_M$, $\bA_2\!=\!\bM_0,\bB_2\!=\!\bM_2$, and $\bA_3\!=\!\bB_3\!=\!\bM_1$. 
	
	Furthermore, denoting by $\widetilde \bc=[\bc^\trans, 1]^\trans \in \bbC^{M+1}$, $f_{\rmI}(\bc)$ can be homogenized as a real-valued homogeneous complex quartic function with respect to $\widetilde \bc$, namely,  
	\beq
	\begin{aligned}\label{f_Hc}
		f_{_\rmH}(\widetilde \bc) =\sum\nolimits_{j=0}^{J}\alpha_j\widetilde \bc^\ctrans\widetilde\bA_j\widetilde\bc\widetilde\bc^\ctrans\widetilde\bB_j^\ctrans\widetilde\bc,
	\end{aligned}
	\eeq
	where $\widetilde \bA_0\!=\!\left[\!\! {\begin{array}{*{20}{c}}
			\disps {\bD}&\bzero\\
			\bzero&0 \!\!\end{array}} \!\!\right]$, $\widetilde \bB_0\!=\!\left[ {\begin{array}{*{20}{c}}
			\bzero& \bzero\\
			\bzero & \disps 1  \end{array}} \right]$, $\widetilde \bA_j\!=\!\left[\!\! {\begin{array}{*{20}{c}}
			\disps \bA_j\!&\!\bzero\\
			\bzero\!&\!0 \!\!\end{array}} \!\!\right]$, and $\widetilde \bB_j\!=\!\left[\!\! {\begin{array}{*{20}{c}}
			\disps \bB_j& \bzero\\
			\bzero & 0  \end{array}} \!\!\right],j=1,2,\cdots,J$.
	With reference to $\widetilde f(\bc)$, $\widetilde \bA_{j_1}$ and $\widetilde \bB_{j_1}$ are Hermitian for $j_1=0,2$ because $\bM_0$ and $\bM_2$ are Hermitian, and $\widetilde \bA_{j_2}$ is equal to $\widetilde \bB_{j_2}$ for $j_2=1,3$. Based on this, $f_{_{\rmH}}(\widetilde \bc)$ can be equivalently expressed as the following conjugate homogeneous quartic function
	\begin{align}\label{f_CH}
		f_{_\text {CH}}(\widetilde \bc) =&\frac{1}{2}\sum\nolimits_{j=0}^{J}\alpha_j\Big[\widetilde \bc^\ctrans(\widetilde\bA_j+\widetilde\bB_j)\widetilde\bc\widetilde\bc^\ctrans(\widetilde\bA_j+\widetilde\bB_j)^\ctrans\widetilde\bc\nonumber\\
		&\quad\quad\quad\quad-\!\widetilde \bc^\ctrans\widetilde\bA_j\widetilde\bc\widetilde\bc^\ctrans\widetilde\bA_j^\ctrans\widetilde\bc\!-\!\widetilde \bc^\ctrans\widetilde\bB_j\widetilde\bc\widetilde\bc^\ctrans\widetilde\bB_j^\ctrans\widetilde\bc\Big].\!\!
	\end{align}
	
	Consequently, according to point \eqref{coro1_i} of Corollary \ref{coro1}, the fourth order conjugate
	super-symmetric tensor associated with $f_{_\text {CH}}(\widetilde \bc)$ (i.e., $\widetilde f(\bc)$) is given by 
	\beq\label{Fpi}
	\!\!\!\!\!\!\!\!\!\!\calF_\pi \!\!=\!\! \left\{\!\!\begin{array}{lll}
		\!\!\!\displaystyle\frac{b_{iqkl}}{\text{Card}(\Pi_4\!(M\!\!+\!\!1\!\!+\!\!i, \!q, \!M\!\!+\!\!1\!\!+\!\!k, \!l))}, & {\begin{array}{l}
				\!\!\!\!\!\!\!\!{\text {if}}\;\pi \!\in\! \Pi_4\!(M\!\!+\!\!1\!\!+\!\!i, \!q, \!M\!\!+\!\!1\!\!+\!\!k, \!l), \\
				\!\!\!\!\!\!\!\!(i,\!q,\!k,\!l) \!\in\! \{1,\!2,\!\cdots,\!M\!\!+\!\!1\}^4,
		\end{array}} \\
		0, & \!\!\!\!\text{otherwise},
	\end{array}
	\right. 
	\eeq
	where, owing to special structure of the involved matrices, $b_{iqkl} =\disps\frac{1}{2}\sum\limits_{j=0}^{J}\sum\limits_{(i',k')\in \Pi_2(i,k)\atop(q',l')\in \Pi_2(q,l)} \!\!\!\!\!\!\!\!\alpha_j\big(\widetilde\bA_j(i',q')\widetilde\bB_j^*(l',k') +\widetilde\bB_j(i',q')\widetilde\bA_j^*(l',k')\big)$, ensuring that $f_{_\text {CH}}(\widetilde \bc)=\widetilde f(\bc)=\calF\big(\binom{\widetilde \bc}{\widetilde \bc^*}, \binom{\widetilde \bc}{\widetilde \bc^*}, \binom{\widetilde \bc}{\widetilde \bc^*}, \binom{\widetilde \bc}{\widetilde \bc^*}\big)$. Moreover, based on point \eqref{coro1_ii} of Corollary \ref{coro1}, the multi-linear function corresponding to $f_{_\text {CH}}(\widetilde \bc)$ can be expressed as 
	\begin{align}
		\!\!\!\!f_{_\text{ML}}&(\widetilde\bc^1, \widetilde\bc^2, \widetilde\bc^3, \widetilde\bc^4)\nonumber\\
		=& \frac{1}{48} \sum\nolimits_{j=0}^{J} \sum\limits_{(p_1, p_2, p_3, p_4) \atop \in \Pi_4 (1, 2, 3, 4)}\!\!\!\!\!\!\!\! \alpha_j\Big((\widetilde\bc^{p_1})^\ctrans\widetilde \bA_j\widetilde\bc^{p_2}(\widetilde\bc^{p_3})^\ctrans\widetilde \bB_j^\ctrans\widetilde\bc^{p_4}\nonumber\\[-0.8em]
		&\quad\quad\quad\quad\quad\quad\quad\quad +(\widetilde\bc^{p_1})^\ctrans\widetilde \bB_j\widetilde\bc^{p_2}(\widetilde\bc^{p_3})^\ctrans\widetilde \bA_j^\ctrans\widetilde\bc^{p_4}\Big)\nonumber\\[-0.4em]
		=&\frac{1}{24} \sum\nolimits_{j=0}^{J} \sum\limits_{(p_1, p_2, p_3, p_4) \atop \in \Pi_4 (1, 2, 3, 4)}\!\!\!\!\!\!\!\! \alpha_j\Big((\widetilde\bc^{p_1})^\ctrans\widetilde \bA_j\widetilde\bc^{p_2}(\widetilde\bc^{p_3})^\ctrans\widetilde \bB_j^\ctrans\widetilde\bc^{p_4}\Big),
	\end{align}
	where $\widetilde \bc^p=[(\bc^p)^\trans, 1]^\trans$, and the second equality holds true since
	\begin{itemize}
		\item if $\widetilde \bA_j=\widetilde \bB_j$, then $(\widetilde\bc^{p_1})^\ctrans\widetilde \bA_j\widetilde\bc^{p_2}(\widetilde\bc^{p_3})^\ctrans\widetilde \bB_j^\ctrans\widetilde\bc^{p_4}=(\widetilde\bc^{p_1})^\ctrans\widetilde \bB_j\widetilde\bc^{p_2}(\widetilde\bc^{p_3})^\ctrans\widetilde \bA_j^\ctrans\widetilde\bc^{p_4}$;
		\item if $\widetilde \bA_j$ and $\widetilde \bB_j$ are Hermitian, then 
		\begin{align}
			&\!\!\!\!\!\!\!\!\!\!\!\!\sum\limits_{(p_1, p_2, p_3, p_4) \atop \in \Pi_4 (1, 2, 3, 4)}\!\!\!\!\!\!\!\!(\widetilde\bc^{p_1})^\ctrans\widetilde \bA_j\widetilde\bc^{p_2}(\widetilde\bc^{p_3})^\ctrans\widetilde \bB_j^\ctrans\widetilde\bc^{p_4}\nonumber\\[-1.2em]
			&\quad\quad\quad=\sum\limits_{(p_1, p_2, p_3, p_4) \atop \in \Pi_4 (1, 2, 3, 4)}\!\!\!\!\!\!\!\!(\widetilde\bc^{p_1})^\ctrans\widetilde \bA_j^\ctrans\widetilde\bc^{p_2}(\widetilde\bc^{p_3})^\ctrans\widetilde \bB_j\widetilde\bc^{p_4}\nonumber\\[-0.4em]
			&\quad\quad\quad=\sum\limits_{(p_1, p_2, p_3, p_4) \atop \in \Pi_4 (1, 2, 3, 4)}\!\!\!\!\!\!\!\!(\widetilde\bc^{p_3})^\ctrans\widetilde \bA_j^\ctrans\widetilde\bc^{p_4}(\widetilde\bc^{p_1})^\ctrans\widetilde\bB_j \widetilde\bc^{p_2},
		\end{align}
		where the last equality exploits the fact that the summation over all permutations of $(p_1, p_2, p_3, p_4)$ in $\Pi_4 (1, 2, 3, 4)$ makes the sums invariant to the order of indices $p_1, p_2, p_3, p_4$.
	\end{itemize}
	
	Finally, recalling the definition of $\widetilde \bA_j,\widetilde \bB_j,j=0,1,\cdots,J$ as reported below equation \eqref{f_Hc}, the multi-linear function $f_{_\text{ML}}(\widetilde\bc^1, \widetilde\bc^2, \widetilde\bc^3, \widetilde\bc^4)$ can be further simplified as
	\begin{align}
		\!\!\!\widetilde f_{_\text{ML}}&(\bc^1, \bc^2, \bc^3, \bc^4) =\frac{1}{12} \sum\limits_{(p_1, p_2) \atop \in \Pi_2 (1, 2, 3, 4)}\!\!\!\!\!\!\!\! \alpha_0(\bc^{p_1})^\ctrans \bD\bc^{p_2}\nonumber\\[-0.5em]
		&+ \frac{1}{24} \sum\nolimits_{j=1}^{J} \sum\limits_{(p_1, p_2, p_3, p_4) \atop \in \Pi_4 (1, 2, 3, 4)}\!\!\!\!\!\!\!\! \alpha_j(\bc^{p_1})^\ctrans \bA_j\bc^{p_2}(\bc^{p_3})^\ctrans \bB_j^\ctrans\bc^{p_4}.
	\end{align}
	
	Finally, specializing the parameters as reported below equation \eqref{fI}, $\widetilde f_{_\text{ML}}(\bc^1, \bc^2, \bc^3, \bc^4)$ can be expressed as in \eqref{f_cccc}.
	
	\subsection{Derivation of Problem \eqref{P_cp}}\label{proofP_inMBI}
	Let us introduce the following lemmas that will be used to obtain the desired expressions.
	\begin{lemma}\label{lemma1}
		{\rm Let matrix $\bD \in \bbH^{M} $, the blocks $\bx^p \in \bbC^M,p=1,2,3,4$, then the restriction of the function $g_1(\bx^1,\bx^2, \bx^3, \bx^4) = \sum\limits_{(p_1, p_2) \atop \in \Pi_2 (1, 2, 3, 4)}(\bx^{p_1})^\ctrans\bD\bx^{p_2}$ to $\bx^p$ with the other blocks fixed, can be recast as $g_1(\bx^p; \{\bx^{\tilde p}\}_{\tilde p \in \calS_{p}})\!= \!\!\!\!\sum\limits_{(p_1, p_2) \atop \in \Pi_2 (\calS_{p})}\!\!\!(\bx^{p_1}\!)^\ctrans\bD\bx^{p_2}\!+\! 2\!\!\!\sum\limits_{p_1 \in \calS_{p}}\!\!\Re\!\left\{\!(\bx^{p_1})^\ctrans\bD\bx^{p}\!\right\}$, where $\calS_p=\{1,2,3,4\}\setminus \{p\}$.}
	\end{lemma}	
	\begin{proof}
		Depending on $p_1$ and $p_2$, $g_1(\bx^1,\bx^2, \bx^3, \bx^4)$ can be expressed as 
		\beq
		\begin{aligned}
			&\sum\limits_{(p_1, p_2) \atop
				\in \Pi_2 (\calS_{p})}\!\!\!(\bx^{p_1})^\ctrans\bD\bx^{p_2}\!+\!\!\!\sum\limits_{p_1 \in \calS_{p}}\!\!\big((\bx^{p})^\ctrans\bD\bx^{p_1}\!+\!(\bx^{p_1}\!)^\ctrans\bD\bx^{p}\big)\!\nonumber\\[-0.5em]
			&\quad\quad =\sum\limits_{(p_1, p_2) \atop
				\in \Pi_2 (\calS_{p})}\!\!\!(\bx^{p_1}\!)^\ctrans\bD\bx^{p_2}\!+\! 2\!\!\!\sum\limits_{p_1 \in \calS_{p}}\!\!\Re\!\left\{\!(\bx^{p_1})^\ctrans\bD\bx^{p}\!\right\}\!.
		\end{aligned}
		\eeq
	\end{proof}
	\begin{lemma}\label{lemma2}
		{\rm For either $\bA , \bB \in \bbH^{M}$ or $\bA= \bB\in\bbC^{M\times M}$, $g_2(\bx^1,\bx^2, \bx^3, \bx^4) =  \sum\limits_{(p_1, p_2, p_3, p_4) \atop \in \Pi_4 (1, 2, 3, 4)}(\bx^{p_1})^\ctrans\bA\bx^{p_2}(\bx^{p_3})^\ctrans\bB^\ctrans\bx^{p_4}$ with $\bx^p \in \bbC^M$, $p=1,2,3,4$, the restriction of $g_2(\bx^1,\bx^2, \bx^3, \bx^4)$ to $\bx^p$ with the other blocks fixed, can be recast as $g_2(\bx^p; \{\bx^{\tilde p}\}_{\tilde p \in \calS_{p}})  = 2\!\!\!\!\!\!\sum\limits_{(p_1, p_2, p_3) \atop \in \Pi_3 (\calS_p)}\!\!\!\!\!\!\Re\!\left\{\!\big(\bA\bx^{p_1}(\bx^{p_2})^\ctrans\!\bB^\ctrans\bx^{p_3}\!+\!\bB^\ctrans\bx^{p_1}(\bx^{p_2}\!)^\ctrans\!\bA\bx^{p_3}\big)^\ctrans\!\bx^p\!\right\}$, where $\calS_p=\{1,2,3,4\}\setminus \{p\}$.}
	\end{lemma}
	\begin{proof}
		$g_2(\bx^1,\bx^2, \bx^3, \bx^4)$ can be expressed as 
		\begin{subequations}\label{g2_xxxx1}
			\begin{align}
				g_2(\bx^1,\bx^2, \bx^3, \bx^4)=&(\bx^{p})^\ctrans\!\!\!\!\!\!\sum\limits_{(p_1, p_2, p_3) \atop \in \Pi_3 (\calS_p)}\!\!\!\!\!\!\bA\bx^{p_1} (\bx^{p_2}\!)^\ctrans\!\bB^\ctrans\bx^{p_3}\label{g2_xxxx_1} \\
				&+  \! \!\!\!\!\!\sum\limits_{(p_1, p_2, p_3) \atop \in \Pi_3 (\calS_p)}\!\!\!\!\!\! (\bx^{p_2}\!)^\ctrans\!\bB^\ctrans\bx^{p_3}(\bx^{p_1}\!)^\ctrans\!\bA\bx^{p}\label{g2_xxxx_2}\\
				&  +\!(\bx^{p})^\ctrans\!\!\!\!\!\!\sum\limits_{(p_1, p_2, p_3) \atop \in \Pi_3 (\calS_p)}\!\!\!\!\!\!\bB^\ctrans\bx^{p_3}  (\bx^{p_1}\!)^\ctrans\!\bA\bx^{p_2}\label{g2_xxxx_3}\\
				&\!+  \! \!\!\!\!\!\!\!\sum\limits_{(p_1, p_2, p_3) \atop \in \Pi_3 (\calS_p)}\!\!\!\!\!\! (\bx^{p_1}\!)^\ctrans\!\bA\bx^{p_2}(\bx^{p_3}\!)^\ctrans\!\bB^\ctrans\bx^{p}\!.\label{g2_xxxx_4}
			\end{align}
		\end{subequations}
		Before proceeding further, let us observe that, $\bx^{p_1}$, $\bx^{p_2}$, $\bx^{p_3}$ with $p_1,$ $p_2$, $p_3\in \calS_p$ can be arbitrary interchanged due to the symmetry of $g_2(\bx^1, \bx^2, \bx^3, \bx^4)$, as explicitly corroborated also by the summations over all permutations $\Pi_3 (\calS_p)$. Now, two cases for $\bA$ and $\bB$ are discussed to further simplify \eqref{g2_xxxx1}. (i) If $\bA , \bB \in \bbH^{M}$, exchanging the labels $p_2$ and $p_3$ in \eqref{g2_xxxx_2}, as well as $p_1$ and $p_2$ in  \eqref{g2_xxxx_4}, results in \eqref{g2_xxxx_1} being conjugate to \eqref{g2_xxxx_2}, and \eqref{g2_xxxx_3} being conjugate to \eqref{g2_xxxx_4}. Then summing over \eqref{g2_xxxx_1}-\eqref{g2_xxxx_4} gives $g_2(\bx^p; \{\bx^{\tilde p}\}_{\tilde p \in \calS_{p}})$. (ii) If $\bA= \bB\in\bbC^{M\times M}$, interchanging $p_1$ and $p_3$ both in \eqref{g2_xxxx_2} and \eqref{g2_xxxx_4} similarly leads to  $g_2(\bx^p; \{\bx^{\tilde p}\}_{\tilde p \in \calS_{p}})$.
	\end{proof}
	
	According to Lemmas \ref{lemma1} and \ref{lemma2}, and plugging in the actual involved parameters, $\widetilde f_{_\text{ML}}(\bc^p; \{\bc^{\tilde p}_{(n)}\}_{\tilde p \in \calS_p})$ can be expressed as 
	\beq
	\begin{aligned}
		\widetilde f_{_\text{ML}}(&\bc^p; \{\bc^{\tilde p}_{(n)}\}_{\tilde p \in \calS_p}) = \Re\big\{(\bd_{(n)}^p)^\ctrans\bc^p\big\}+d_{(n)}^p,\nonumber\\
	\end{aligned}
	\eeq
	where
	\begin{subequations}\label{dnp}
		\begin{align}
			\!\!\!\!\!\!&\bd_{(n)}^p= \frac{1}{12} \!\!\!\sum\limits_{(p_1, p_2, p_3) \atop \in \Pi_3 (\calS_p)}\!\!\!\!\Big\{2\mu_1 \bc_{(n)}^{p_1}(\bc_{(n)}^{p_2})^\ctrans\bc_{(n)}^{p_3}\nonumber \\[-0.3em]
			& \quad\!+ \! (1\!-\!\beta)\big(\bM_{\!0} \bc_{(n)}^{p_1}(\bc_{(n)}^{p_2}\!)^\ctrans\!\bM_{\!2} \bc_{(n)}^{p_3}\!+\!\! \bM_{\!2}\bc_{(n)}^{p_1}(\bc_{(n)}^{p_2}\!)^\ctrans\!\bM_{\!0}\bc_{(n)}^{p_3}\big)\!\nonumber \\
			& \quad\!- \! (1\!-\!\beta)\big(\bM_{\!1}^\ctrans\bc_{(n)}^{p_1}\! (\bc_{(n)}^{p_2}\!)^\ctrans\!\bM_{\!1}\bc_{(n)}^{p_3}\!\!+\!\! \bM_{\!1} \bc_{(n)}^{p_1}\!(\bc_{(n)}^{p_2}\!)^\ctrans\!\bM_{\!1}^\ctrans \bc_{(n)}^{p_3}\big) \!\Big\}\nonumber \\
			& \quad\!+\! \frac{1}{6} \sum\limits_{p_1\in \calS_p} \big(\mu_2\bc_{(n)}^{p_1}+ \beta \bM_{\!0}\bc_{(n)}^{p_1}\big),\label{dd_p}\\
			&d_{(n)}^p = \disps \frac{1}{12} \!\!\!\!\sum\limits_{(p_1, p_2) \atop \in \Pi_2 (\calS_p )}\!\!\!\!\left(\mu_2(\bc_{(n)}^{p_1})^\ctrans\bc_{(n)}^{p_2} + \beta (\bc_{(n)}^{p_1})^\ctrans\bM_0\bc_{(n)}^{p_2}\right)\!.\!\!\label{d_p}
		\end{align}
	\end{subequations}
	\subsection{Proof of Lemma \ref{lemma3_1}}\label{prooflemma3_1}	
	Note that Problem \eqref{P_cp2} is a convex optimization problem with a non-empty and compact set, thus, it has at least one optimal solution. The main idea of the proof is to establish the existence of an optimal solution to the relaxed Problem \eqref{P_cp2} that satisfies $\|\bc^{p}\|^2\!=\!1$ in Problem \eqref{P_cp1}. To verify this, it is enough to show that for any feasible point $\breve \bc^p$ with $2\Re\{\bc_0^\ctrans\breve\bc^p\}\!=\!2\!-\!\zeta_1, 0\!\leq\!\zeta_1\!\leq\! \zeta$, of Problem \eqref{P_cp2}, it is always possible to find a point with unit energy and $2\Re\{\bc_0^\ctrans\bc^p\}=2-\zeta_1$ leading to an objective value greater than or equal to that of $\breve \bc^p$.	To this end, let us consider $\bc^p_i \!=\! \breve\bc^p \!+\! \sqrt{1\!-\!\|\breve\bc^p\|^2}\disps {(-1)^i\bc_{0}^{\perp}}/{\|\bc_{0}^{\perp}\|},i=1,2$, where $\bc_{0}^{\perp}$ is a non-zero orthogonal vector of $\bc_0$ and $\breve\bc^p$. It is clear that $\bc^p_i$ has unit energy and fulfills $2\Re\{\bc_0^\ctrans\bc_i^p\}=2-\zeta_1$. Besides, either for $i= 1$ or $i = 2$, the following inequalities can be valid $\widetilde f_{_\text{ML}}(\bc^p_i; \{\bc^{\tilde p}_{(n)}\}_{\tilde p \in \calS_p})- \widetilde f_{_\text{ML}}(\breve \bc^p; \{\bc^{\tilde p}_{(n)}\}_{\tilde p \in \calS_p})
	=\sqrt{1-\|\breve\bc^p\|^2}\disps {(-1)^i\Re\big\{( \bd_{(n)}^p)^\ctrans \bc_{0}^{\perp}\big\}}/{\|\bc_{0}^{\perp}\|}\geq 0$ depending on sign of $\Re\big\{( \bd_{(n)}^p)^\ctrans \bc_{0}^{\perp}\big\}$.
	\subsection{Closed-form Solution to Problem \eqref{P_cp2}}\label{proofsolutionP_cp}
	
	The KKT conditions for the convex Problem \eqref{P_cp2} are 
	\begin{subequations}\label{KKTs}
		\begin{eqnarray}
			1-{\left\| {\bc^p} \right\|^2} \ge 0,  &&\label{energy}\\
			{\lambda _1}( {1-{{\left\| {\bc}^p \right\|}^2}  } ) = 0,\label{optlam1}&&\\
			{\lambda _1} \ge 0,&&\label{conlam1}\\
			2\Re \big\{ {{\bc}_0^\ctrans{\bc^p}} \big\} -2+ \zeta \geq 0,&&\label{similar}\\
			{\lambda _2}( 2\Re \big\{ {{\bc}_0^\ctrans{\bc^p}} \big\} -2+ \zeta ) = 0,\label{optlam2}&&\\
			{\lambda _2} \ge 0,&&\label{conlam2}\\
			\bd_{(n)}^p - 2\lambda_1\bc^{p} + 2\lambda_2\bc_0 = \bzero_M,\label{opt1}&&
		\end{eqnarray}
	\end{subequations}
	where $\lambda_1$ and $\lambda_2$ are the Lagrange multipliers. According to Lemma \ref{lemma3_1}, the solution with unit energy, satisfying the KKT conditions in \eqref{KKTs} exists. In the following, four cases for the value of  $\bd^p_{(n)}$ are discussed to obtain the optimal solution. 
	\begin{enumerate}[(a)]
		\item \label{E_case1} If  $\bd_{(n)}^p\!=\!\bzero_M$, the optimal solution can be given by $\bc^{p}\!=\!\bc_0$ because it fulfills the KKT conditions with $\lambda_1\!=\!\lambda_2\!=\!0$; otherwise, consider case \eqref{E_case2}.
		\item \label{E_case2} If $\bd_{(n)}^p\!=-\varrho\bc_0, \varrho>0$, then any point $\bc^{p}={(2-\zeta)} \bc_0/{2} + \sqrt{1-{(2-\zeta)}^2/{4}}{\bc_{0}^{\perp}}/{\|\bc_0^{\perp}\|}$	is an optimal solution with unit norm, because it satisfies $2\Re\{\bc_0^\ctrans\bc^{p}\}=2\!-\!\zeta$ and also the KKT conditions with $\lambda_1=0, \lambda_2=\varrho/2>0$; otherwise, jump to case \eqref{E_case3}. 
		\item\label{E_case3} If $2\Re \big\{ {\bc}_0^\ctrans\bd_{(n)}^p/\|\bd_{(n)}^p\!\| \big\}\!-\!2\!+ \!\zeta \!\geq\! 0$, then $\bc^{p}\!=\!\bd_{(n)}^p/\|\bd_{(n)}^p\!\|$ is the desired optimal solution since it meets the KKT conditions with $\lambda_1\!\!\!>\!\!0, \lambda_2\!\!=\!\!0$; otherwise, consider case \eqref{E_case4}.
		\item\label{E_case4} Following the KKT conditions, the closed-form solution\footnote{In this case,  $\bd_{(n)}^p+2\lambda_2\bc_0\neq\bzero_M$ always holds, otherwise the solution is provided in case \eqref{E_case2}.} $\bc^{p} = \frac{\bd_{(n)}^p+2\lambda_2\bc_0}{\|\bd_{(n)}^p+2\lambda_2\bc_0\|}$ is the optimal solution to Problem \eqref{P_cp2} as long as a non-negative $\lambda_2$ can be found that enables
		\beq\label{similar1}
		2\Re \big\{ {\bc}_0^\ctrans\bc^p \big\}-2+ \zeta= 0.
		\eeq
		Now, let us prove that such $\lambda_2$ can be found. By substituting this $\bc^{p}$ into \eqref{similar1}, it follows that
		\beq \label{eq_lam2}
		2{\Re \big\{ {{\bc}_0^\ctrans(\bd_{(n)}^p+2\lambda_2\bc_0)} \big\}} =(2- \zeta)\|\bd_{(n)}^p+2\lambda_2\bc_0\|.
		\eeq
		Squaring both sides of the equation \eqref{eq_lam2} and using $\|\bc_0\|^2\!=\!1$ gives 
		\beq\label{21lambda2}
		a_{(n)}^p\lambda_2^2 + b_{(n)}^p\lambda_2 + c_{(n)}^p = 0,
		\eeq
		where $\!a_{(n)}^p=\!4[(2-\zeta)^2-4], b_{(n)}^p=\!4\Re\{\bc_0^\ctrans\bd_{(n)}^p\}[(2\!-\!\zeta)^2\!-\!4], c_{(n)}^p\!=\!(2\!-\!\zeta)^2\|\bd_{(n)}^p\|^2-4\big(\Re\{\bc_0^\ctrans\bd_{(n)}^p\}\big)^2$. 
		For\footnote{If $\zeta=0$, the unique feasible point $\bc_0$ is the optimal solution to Problem \eqref{P_cp2}, since the KKT conditions hold true for any $\lambda_2\geq 0$.} $0<\zeta\!\leq\! 2$, $a_{(n)}^p \!<\! 0$, and  $c_{(n)}^p \!>\! 0$ since $2\Re \big\{ {{\bc}_0^\ctrans\bd_{(n)}^p} \big\}\!<\! \|\bd_{(n)}^p\|(2\!-\!\zeta)$ (otherwise the optimal solution is provided by case \eqref{E_case3}). According to the quadratic formula ($a_{(n)}^p<0$), equation \eqref{21lambda2} has two real roots, and the positive one is 
		\beq\label{pos_lam2}
		\!\!\!\lambda_2\!=\!\disps\frac{1}{2}\left(-\! {b_{(n)}^p}/{a_{(n)}^p} \!+\! \sqrt{({b_{(n)}^p}/{a_{(n)}^p})^2 \!-\!4{c_{(n)}^p}/{a_{(n)}^p}}\right).
		\eeq
	\end{enumerate}
	
	\subsection{Proof of Proposition \ref{Propo2}}\label{proof_propo2}
	A corollary for the convexity analysis of a real-valued function is provided\cite{boyd2004convex}. 
	\begin{corollary}\label{corollary1}
		{\rm Given a real-valued function $ \widetilde g\!: \bbC^M\!\rightarrow\!\bbR$, $\bx_0\!\in\! \bbC^M$ and $\by_0\!\in\! \bbC^M$, define the restriction $\psi_{(\bx_0,\by_0)}\!:\bbR\rightarrow\bbR$ as $\psi_{(\bx_0,\by_0)}(t)\!=\!\widetilde g(\bx_0+t\by_0)$. Then, $\widetilde g$ is convex on $\bbC^M$ if and only if for any $\bx_0, \by_0 \in \bbC^M$, $\psi_{(\bx_0,\by_0)}(t)$ is convex on $\bbR$; $\psi''_{(\bx_0,\by_0)}(t)\geq 0$ for all $t\in \bbR$, assuming that the second order derivatives exist.} 
	\end{corollary}
	
	Based on the properties of $\widetilde\bA_j$, $\widetilde \bB_j$ (same or Hermitian) for $j=1,\cdots,J$, let us first express $\widetilde g(\bx)\!=\!\mu_1(\bx^\ctrans\bx)^2\!+\!\mu_2(\bx^\ctrans\bx)\!+\!g_{_\rmI}(\bx)$ equivalently as
	\begin{align}
		\widetilde g(\bx) =&\mu_1(\bx^\ctrans\bx)^2\!+\!\mu_2(\bx^\ctrans\bx)+\alpha_0\bx^\ctrans\bD\bx\nonumber\\
		&+\frac{1}{2}\sum\nolimits_{j=1}^{J}\alpha_j\Big[ \bx^\ctrans(\bA_j+\bB_j)\bx\bx^\ctrans(\bA_j+\bB_j)^\ctrans\bx\nonumber\\
		&\quad\quad\quad\quad-\! \bx^\ctrans\bA_j\bx\bx^\ctrans\bA_j^\ctrans\bx\!-\! \bx^\ctrans\bB_j\bx\bx^\ctrans\bB_j^\ctrans\bx\Big].\!\!\!\!
	\end{align}
	
	According to Corollary \ref{corollary1}, to prove that $\widetilde g(\bx)$ is convex, it is necessary to show that $\psi_{(\bh,\by)}(t)\!=\!\widetilde g(\bh+t\by), \forall \bh, \by \in \bbC^M$ is convex. Leveraging the results for conjugate homogeneous quartic function in \cite[Theorem 3.3]{aubry2013ambiguity}, and denoting by $\bz=\bh+t\by, \forall \bh, \by \in \bbC^M$, the second order derivatives of $\psi_{(\bh,\by)}(t)$ with respect to $t$ can be expressed as
	\beq
	\begin{aligned}
		\psi''_{(\bh,\by)}(t)=& 2\mu_1(2\by^\ctrans\by\bz^\ctrans\bz+(\by^\ctrans\bz+\bz^\ctrans\by)^2)\\
		& + 2\by^\ctrans(\mu_2\bI_M+\alpha_0\bD)\by \\
		&+\sum\nolimits_{j=1}^{J}\alpha_j\Big(h( \by,  \bz,  \bA_j\!+\! \bB_j,  \bA_j\!+\! \bB_j)\\
		&\quad\quad\quad\;\;-\!h( \by,  \bz,  \bA_j,  \bA_j)\!-\!h( \by,  \bz,  \bB_j,  \bB_j)\Big),
	\end{aligned}
	\eeq
	where $h(\by, \bz, \cdot, \cdot)$ is given in \eqref{h_yz}. By resorting to the properties of $\widetilde\bA_j$, $\widetilde \bB_j$ again, one can obtain $h( \by,  \bz,  \bA_j\!+\! \bB_j,  \bA_j\!+\! \bB_j)-\!h( \by,  \bz,  \bA_j,  \bA_j)\!-\!h( \by,  \bz,  \bB_j,  \bB_j)=2h( \by,  \bz,  \bA_j,  \bB_j)$.	Then, $\psi''_{(\bh,\by)}(t)$ can be simplified as $\psi''_{(\bh,\by)}(t)= 2\mu_1(2\by^\ctrans\by\bz^\ctrans\bz+(\by^\ctrans\bz+\bz^\ctrans\by)^2) + 2\by^\ctrans(\mu_2\bI_M+\alpha_0\bD)\by +2\sum\limits_{j=1}^{J}\alpha_j h( \by,  \bz,  \bA_j,  \bB_j)$. 
	
	Hence, if for any $\by,\bz\in\bbC^M$, the conditions 
	\begin{align}
		\!\!\!\!\mu_1(2\by^\ctrans\by\bz^\ctrans\bz\!+\!(\by^\ctrans\bz\!+\!\bz^\ctrans\by)^2)\!+\!{\sum\nolimits_{j=1}^{J}\!\alpha_jh(\by, \bz, \bA_j, \bB_j)}\!\geq\!0, \nonumber\\
		\!\!\!\!\!\!\by^\ctrans(\mu_2\bI_M\!+\!\alpha_0\bD)\by\!\geq\!0,\nonumber
	\end{align}
	hold true, then $\psi''_{(\bh,\by)}(t)\!\geq\!0$, that is, $\psi_{(\bh,\by)}(t)$ is convex. The first inequality holds for all $\by, \bz \!\in\! \bbC^M$ when $\mu_1\!\geq\!\mu_1^\star$, where $\mu_1^\star\!=\!\!\arg \max\limits_{\|\by\|^2=1,\|\bz\|^2=1} {- \frac{1}{2}\sum\limits_{j=1}^{J}\alpha_j h(\by, \bz, \bA_j, \bB_j)}$ resorting to \cite[Corollary 3.4]{aubry2013ambiguity}. The second inequality holds true if $\mu_2\bI_M+\alpha_0\bD \succeq 0$, namely, $\mu_2 \!\geq \max \{0, -\lambda_{\min}(\alpha_0\bD)\}$. 
	
	\subsection{Derivation of \eqref{h_up}}\label{proofupbound}
	Denoting by $\bu=[\by^\trans, \bz^\trans]^\trans\in \bbC^{2M}$, let us first represent $\widetilde h(\by, \bz)= h(\by, \bz, \bM_1, \bM_1)- h(\by, \bz, \bM_0, \bM_2)$ as an equivalent function $\breve h(\bu)$ as follows
	\beq
	\breve h(\bu) =  h_1(\bu, \bM_1, \bM_1) -  h_2(\bu, \bM_0, \bM_2),
	\eeq
	where $h_1(\bu, \bM_1, \bM_1) \!=\! |\bu^\ctrans \bM_{1, 1} \bu|^2  \!-\! |\bu^\ctrans \bM_{1, 2} \bu|^2 
	\!-\! |\bu^\ctrans \bM_{1, 3} \bu|^2 \!+\! |\bu^\ctrans \bM_{1, 4} \bu|^2$ 	with $\bM_{1, q_1}\in \bbC^{2M\times 2M}, q_1=1,2,3,4$ given by
	
	\beq 
	\begin{aligned}
		\!\!\!\!\!\!\bM_{1, 1}&\!=\!\!\left[\!\! {\begin{array}{*{20}{c}}
				\!\!\bM_1\!\!\!&\!\!\!\bzero\\
				\bzero\!\!\!&\!\!\!\bM_1 \!\!\end{array}} \!\!\right]\!\!,   \bM_{1, 2}\!=\!\!\left[\!\! {\begin{array}{*{20}{c}}
				\bM_1&\bzero\!\!\\
				\bzero&\bzero\!\! \end{array}} \!\!\right]\!\!,\\
		\bM_{1, 3}&\!=\!\!\left[\!\! {\begin{array}{*{20}{c}}
				\bzero&\bzero\!\!\\
				\bzero&\bM_1\!\! \end{array}} \!\!\right]\!\!, \bM_{1, 4}\!=\!\!\left[\!\! {\begin{array}{*{20}{c}}
				\!\!\bzero\!\!\!&\!\!\!\bM_1\!\!\\
				\!\!\bM_1\!\!\!&\!\!\!\bzero \!\!\end{array}} \!\!\right]\!\!,
	\end{aligned}
	\eeq
	and $h_2(\bu, \bM_0, \bM_2) 
	=  \bu^\ctrans \bM_{0, 1} \bu \bu^\ctrans \bM_{2, 1}^\ctrans \bu + \bu^\ctrans \bM_{0, 2} \bu $ $\bu^\ctrans \bM_{2, 2}^\ctrans \bu + \bu^\ctrans \bM_{0, 3} \bu \bu^\ctrans \bM_{2, 3}^\ctrans \bu$ with  $\bM_{0, q_2}, \bM_{2, q_2}\in \bbH^{2M\times 2M}, q_2=1,2,3$ given by
	\begin{align}
		&\bM_{0, 1}\!=\!\!\left[\! {\begin{array}{*{20}{c}}
				\!\!\bM_0\!&\!\bzero\\
				\bzero\!&\!\bzero \end{array}} \!\!\right]\!\!,   \bM_{0, 2}\!=\!\!\left[\!\!\! {\begin{array}{*{20}{c}}
				\bzero\!\!\!&\!\!\!\bM_0\\
				\!\bM_0\!\!\!&\!\!\!\bzero \end{array}} \!\!\!\right]\!\!, \bM_{0, 3}\!=\!\!\left[\! {\begin{array}{*{20}{c}}
				\!\!\bzero\!&\!\bzero\\
				\!\!\bzero\!&\!\bM_0\!\! \end{array}} \!\right]\!\!, \nonumber \\
		\!\!\!\!\!\!& \bM_{2, 1}\!=\!\!\left[ \!{\begin{array}{*{20}{c}}
				\!\!\bzero\!&\!\bzero\!\!\\
				\!\!\bzero\!&\!\bM_2\!\! \end{array}} \!\right]\!\!, \bM_{2, 2}\!=\!\!\left[ \!\!\!{\begin{array}{*{20}{c}}
				\bzero\!\!\!&\!\!\!\bM_2\!\!\\
				\bM_2\!\!\!&\!\!\!\bzero\!\! \end{array}} \!\!\!\right]\!\!,\bM_{2, 3}\!=\!\!\left[\! {\begin{array}{*{20}{c}}
				\!\!\bM_2\!&\!\bzero\!\!\\
				\!\!\bzero\!&\!\bzero\!\! \end{array}} \!\right]\!\!.
	\end{align}
	
	Now, a lower bound (non-tight) for $ h_2(\bu, \bM_0, \bM_2)$ is discussed by analyzing the range of the eigenvalues of the involved Hermitian matrices. The matrices $\bM_{0, q_2}\succeq 0, \bM_{2, q_2}\succeq 0, q_2=1,3$, because of $\bM_0\succ 0$ and $\bM_2\succ \!0$. Based on $\text{eig}(\bA\otimes \bB)\!=\!\text{eig}(\bA)\otimes \text{eig}(\bB)$\cite{seber2008matrix}, the eigenvalues of $\text{eig}(\bM_{0,2})$ and $\text{eig}(\bM_{2,2})$ can be computed as $\text{eig}(\bM_{0,2})\!=\![1, -1]^\trans\!\otimes \!\text{eig}(\bM_{0})$ and $\text{eig}(\bM_{2,2})\!=\![1, -1]^\trans\!\otimes \!\text{eig}(\bM_{2})$, respectively. This implies that the eigenvalues of $\bM_{0,2}$ lies in the interval  $\big[\!-\!\lambda_{\max}(\bM_0), \lambda_{\max}(\bM_0)\big]$, and the eigenvalue of $\bM_{2,2}$ belongs to $\big[\!-\!\lambda_{\max}(\bM_2), \lambda_{\max}(\bM_2)\big]$. Hence, a lower bound for $ h_2(\bu, \bM_0, \bM_2)$ is $h_2(\bu, \bM_0, \bM_2) > -{\lambda_{\max}}(\bM_0) {\lambda_{\max}}(\bM_2)\|\bu\|^4$.
	
	Therefore, for any $\bu \in \bbC^{2M}$, $\breve h(\bu)$ satisfies 
	\beq\label{up_h1u}
	\!\!\!\!\breve h(\bu) \!<\! \sum\nolimits_{q_1=1}^4w_{q_1}|\bu^\ctrans \bM_{1, q_1} \bu|^2 \!+\! {\lambda_{\max}}(\bM_0) {\lambda_{\max}}(\bM_2)\|\bu\|^4,\!\!\!\!
	\eeq
	where $w_1\!=\!w_4\!=\!1,w_2\!=\!w_3\!=\!-1$. Letting $\bU\!=\!\bu\bu^\ctrans$, the RHS of equation \eqref{up_h1u} is recast as  $\vect(\bU)^\ctrans\bPhi \vect(\bU)$, where 
	\beq
	\begin{aligned} \label{Phi}
		\bPhi=& {{\lambda_{\max}}(\bM_0) {\lambda_{\max}}(\bM_2)}\bI_{4M^2} \\
		&+  \sum\nolimits_{q_1=1}^4 w_{q_1} \vect\!(\bM_{1, q_1})\vect\!(\bM_{1, q_1})^\ctrans. 
	\end{aligned} 
	\eeq
	Therefore, an upper bound of $\breve h(\bu)$ is given by\cite{sun2016majorization}	
	\beq\label{up_bu2}
	\breve h(\bu) < \lambda_{\max}(\bPhi) \vect(\bU)^\ctrans\vect(\bU)=\lambda_{\max}(\bPhi) \|\bu\|^4.
	\eeq
	
	Substituting $\bu=[\by^\trans, \bz^\trans]^\trans$ back into \eqref{up_bu2}, yields $\breve h(\bu)=\widetilde h(\by, \bz)<\lambda_{\max}(\bPhi) (\|\by\|^2+\|\bz\|^2)^2$.
	
	\bibliography{E:/1Reasearch/4Paper/Refs_dataset/TaoWaveformRefs}
	
\end{document}